\documentclass[11pt]{amsart}
\usepackage{amsmath,amssymb,amsfonts,amscd,epsfig,axodraw4j,pstricks,color}

\theoremstyle{plain}
\newtheorem{thm}{Theorem}
\newtheorem{lem}[thm]{Lemma}
\newtheorem{con}[thm]{Conjecture}
\newtheorem{cor}[thm]{Corollary}

\newtheorem{remark}[thm]{Remark}
\newtheorem{defn}[thm]{Definition}
\newtheorem{ex}[thm]{Example}

\newtheorem{result}[thm]{Result}
\newtheorem{problem}[thm]{Problem}
\newcommand{\dd}{{\rm d}}

\newcommand{\sgn}{{\rm sgn}}

\newcommand{\sE}{{\mathcal E}}

\newcommand{\sM}{{\mathcal M}}

\newcommand{\sV}{{\mathcal V}}

\newcommand{\FF}{{\mathbb F}}

\newcommand{\QQ}{{\mathbb Q}}

\newcommand{\ZZ}{{\mathbb Z}}

\renewcommand\[{\begin{equation}}
\renewcommand\]{\end{equation}}
\newcommand\A{{\bf A}}
\renewcommand\P{{\bf P}}
\renewcommand\tilde{\widetilde}
\renewcommand\phi{\varphi}
\renewcommand\epsilon{\varepsilon}

\renewcommand\theta{\vartheta}
\renewcommand\rho{\varrho}
\newcommand\fL{{\mathfrak L}}

\newcommand\fX{{\mathfrak X}}
\newcommand\fY{{\mathfrak Y}}
\newcommand\cL{{\mathcal L}}
\newcommand\cO{{\mathcal O}}
\newcommand\Fq{{\overline{f}}}
\newcommand\Xq{{\overline{X}}}
\newcommand\into{\hookrightarrow}
\newcommand\auf{\twoheadrightarrow}

\title{Geometries in perturbative quantum field theory}
\author{Oliver Schnetz}
\address{Friedrich Knop\\
Department Mathematik\\
Cauerstra{\ss}e 11\\
91058 Erlangen, Germany}
\email{knop@mi.uni-erlangen.de}
\address{Oliver Schnetz\\
Department Mathematik\\
Cauerstra{\ss}e 11\\
91058 Erlangen, Germany}
\email{schnetz@mi.uni-erlangen.de}

\begin{document}

\begin{abstract}
In perturbative quantum field theory one encounters certain, very specific geometries over the integers.
These perturbative quantum geometries determine the number contents of the amplitude considered.
In the article `Modular forms in quantum field theory' F. Brown and the author report on a first list of perturbative quantum geometries
using the $c_2$-invariant in $\phi^4$ theory. A main tool was denominator reduction which allowed
the authors to examine graphs up to loop order (first Betti number) 10.

We introduce an improved quadratic denominator reduction which makes it possible to extend the previous results to loop order 11
(and partially orders 12 and 13). For comparison, also non-$\phi^4$ graphs are investigated. Here, we extend the results from loop order 9 to 10.
The new database of 4801 unique $c_2$-invariants (previously 157)---while being consistent with all major $c_2$-conjectures---leads to a more refined picture
of perturbative quantum geometries.

In the appendix, Friedrich Knop proves a Chevalley-Warning-Ax theorem for double covers of affine space.
\end{abstract}
\maketitle

\section{Introduction}
For any connected graph $G$ the graph polynomial is defined by associating a variable $\alpha_e$ to every edge $e$ and setting
\begin{equation}\label{1}
\Psi_G(\alpha)=\sum_{T\,\rm span.\,tree}\;\prod_{e\not \in T}\alpha_e,
\end{equation}
where the sum is over all spanning trees $T$ of $G$. These polynomials first appeared in Kirchhoff's work on currents in electrical networks \cite{KIR}.
The graph polynomial is linked by a matrix tree theorem to the determinant of the graph Laplacian (which we use in Definition \ref{defpsi}).
We write $\sE(G)$ and $\sV(G)$ for the set of edges and vertices in $G$, respectively. The loop order of $G$ is its first Betti number $h_1(G)$.

The degree of a vertex in $G$ is the number of (half-)edges incident to $v$. We say that
\begin{equation}
G\hbox{ is in }\phi^4\hbox{ (theory) if }\deg(v) \leq 4\hbox{ for all vertices } v\in\sV(G).
\end{equation}
We also use the term valence for the degree. A graph that is not in $\phi^4$ is a non-$\phi^4$ graph.

The arithmetic contents of perturbative $\phi^4$ (quantum field) theory is given by integrals over rational functions whose denominators are the squares of graph polynomials
\cite{Scensus},
\begin{equation}\label{2}
P(G)=\int_0^\infty\cdots\int_0^\infty\frac{{\rm d}\alpha_1\cdots{\rm d}\alpha_{|\sE(G)|-1}}{\Psi_G(\alpha)^2|_{\alpha_{|\sE(G)|}=1}}, 
\end{equation}
whenever the integral exists. A graph with existing period is called primitive (see Definition \ref{defprim} for a graph theoretical description).
In $\phi^4$ theory the period $P$ is a renormalization scheme independent contribution to the $\beta$-function \cite{IZ}. Interestingly, also the quantum electrodynamical
contribution to the anomalous magnetic moment of the electron shows closely related arithmetic contents \cite{Sg2}.

Since $(\ref{1})$ is defined over the integers, it defines an affine scheme of finite type over $\mathrm{Spec}\,\ZZ$ which is called the graph hypersurface $X_G\subset\A^{|\sE(G)|}$.
For any field $k$, we can therefore consider the zero locus $X_G(k)$ of $\Psi_G$ in $k^{|\sE(G)|}$. If the ground field $k\cong\FF_q$ is finite, we have the point-counting function
\begin{equation}\label{2a}
[X_G]_q:q\mapsto|X_G(\FF_q)|.
\end{equation}
It defines a map from prime powers to non-negative integers. Inspired by the appearance of multiple zeta values in the period integral,
Kontsevich informally conjectured in 1997 that the function $[X_G]$ might be polynomial in $q$ for all graphs \cite{Kconj}.

Although the conjecture is false in general \cite{BB,SFq,Dexample}, a connection between the point-counting function and the period $(\ref{2})$ remains valid.
Certain information about the period is indeed detected by a small piece of the point-counting function, called the $c_2$-invariant, see \cite{SFq,K3,BD}.
For every graph $G$ with at least three vertices $[X_G]_q$ is divisible by $q^2$ (Theorem 2.9 in \cite{SFq}). For these graphs we can define
\begin{equation}\label{c2def}
c_2^{(q)}(G)\equiv\frac{[X_G]_q}{q^2} \mod q.
\end{equation}
For a fixed graph $G$ the $c_2$ maps prime powers $q$ to residues modulo $q$.

In the case when $[X_G]_q$ is a polynomial in $q$, the $c_2$-invariant is the reduction modulo $q$ of the coefficient of $q^2$ in this polynomial.
The connection between the period and the $c_2$-invariant is further borne out by the following conjecture, which holds in all presently known examples.
\begin{con}[Remark 2.11 (2) in \cite{SFq}, Conjecture 2 in \cite{BSmod}]\label{con2}
If $P(G_1)=P(G_2)$ for primitive graphs $G_1$, $G_2$, then $c_2^{(q)}(G_1)\equiv c_2^{(q)}(G_2)\mod q$ for all prime powers $q$.
\end{con}
The conjecture is supported by \cite{BD}, where it is shown that, for a large class of graphs, the $c_2$-invariant is related to the de Rham framing on the cohomology of the graph hypersurface
given by the integrand of (\ref{2}). In this sense the $c_2$ detects geometries (motives) in quantum field theory.

Still, the $c_2$ is only a very crude approximation to the geometry of the Feynman integral (\ref{2}). For the full information one has to construct the entire motive of the
pair given by the domain of integration and the differential form. The importance of the $c_2$ relies on the evidence that in many setups it is able to detect the most complicated
part of the motive \cite{BD}. In this article we exclusively look at the $c_2$ and reduce the term (perturbative) quantum geometry to the section of the geometry which is seen by the $c_2$.

It is often practical to only consider primes instead of prime powers. We conjecture that no information is lost by this reduction.
\begin{con}\label{conqp}
Let $G_1$ and $G_2$ be two graphs with $c_2^{(p)}(G_1)\equiv c_2^{(p)}(G_2)\mod p$ for all primes $p$. Then $c_2^{(q)}(G_1)\equiv c_2^{(q)}(G_2)\mod q$ for all prime powers $q=p^n$.
\end{con}
In particular, the identification of $c_2$-invariants via modular forms in Definition \ref{defmod} uses only primes. Point-counts over finite non-prime fields ($q=p^n,\,n\geq2$) are
significantly more time consuming than over fields $\FF_p$ where simple integer arithmetic can be used. Still, these point counts are possible for small prime powers, see \cite{Stem}, and the
above conjecture could be tested. We are not aware of any progress in this direction beyond \cite{SFq}, so that the above conjecture has only thin experimental support.

At low loop orders there exist only few, rather trivial $c_2$-invariants (notably $-1$ and 0). At higher loops the picture becomes much more diverse.
The graph polynomial that defines the $c_2$ grows rapidly in size with the loop order. Therefore it is desirable to develop tools beyond point-counting $X_G$ to determine the $c_2$ at
least for small primes.

In \cite{K3} it was shown that the $c_2$ can be obtained by point-counting the product of two smaller {\em Dodgson} polynomials with no division by $q^2$, see
Definition \ref{defpsi} and Equation (\ref{c2dod}). But simplification does not stop there. If the $c_2$ is given by point-counting a product of two polynomials, each of which is
linear in an (arbitrarily chosen) variable $\alpha$, one may take resultants with respect to $\alpha$,
$$
[(A\alpha+B)(C\alpha+D)]_q\equiv -[AD-BC]_q\mod q.
$$
The above equation holds if the total degree of $(A\alpha+B)(C\alpha+D)$ does not exceed the number of variables. In Theorem \ref{denredthm} this is ensured by
the condition $2h_1(G)\leq|\sE(G)|$. This tool (denominator reduction) can be iterated to quite efficiently reduce the polynomials that need to be counted.
It enabled the authors of \cite{BSmod} to perform an exhaustive empirical search for $c_2$s that come from primitive graphs with at most 10 loops.

Some $c_2$s are linked to the Fourier coefficients of certain modular forms (see \cite{K3,BSmod}). A striking result was the
many gaps in the geometries found, see Table \ref{tab1}. In particular, there was no modular form of weight 2 and level $\leq1200$ which matched the $c_2$ of any primitive graph of $\leq10$ loops
in $\phi^4$ theory (Conjecture 26 in \cite{BSmod}, see Conjecture \ref{nowt2}). This `no curves conjecture' does not eliminate curves from perturbative quantum field theory. It refers
to the situation that a Feynman integral (of the type considered in this article) has a representation as an elliptic integral in one variable with a numerator given by hyperlogarithms.
This means that after a series of hyperlogarithmic integrations only the last integral is elliptic. In quantum field theory it is possible that elliptic integrals are iterated.
One prominent example is the K3 surface in \cite{K3} (corresponding to the weight 3 level 7 modular form) which is the symmetric square of an elliptic curve.
An example of a different type is the Feynman integral of the massive sunrise diagram (see e.g.\ \cite{BVsunrise} and the references therein).

The empirical evidence in \cite{BSmod} was based on only 157 unique $c_2$s. Higher loop orders, however, could not be examined because even after denominator reduction
the point-counting was too time-consuming.

To overcome this difficulty, we prove a more powerful version of denominator reduction for odd prime powers.
We express any point-count as a sum over Legendre symbols (Definition \ref{deflege}). For $a\in\FF_q$ we define
$$
\left(\frac{a}{q}\right)=\left\{\begin{array}{cl}1&\hbox{, if }a\neq0\hbox{ is a square in }\FF_q,\\
0&\hbox{, if }a=0,\\
-1&\hbox{, otherwise.}\end{array}\right.
$$
For any polynomial $F$ in a positive number $N$ of variables we have
$$
-[F]_q\equiv(F^2)_q:=\sum_{\alpha\in\FF_q^N}\left(\frac{F(\alpha)^2}{q}\right)\mod q.
$$
A Legendre sum $(X)_q$ translates to a denominator $\sqrt{X}$ in (\ref{2}), see Section \ref{sectdr}.
Integration over square roots suggests a more powerful {\em quadratic} denominator reduction with two cases (see Definition \ref{defqdr} and Theorem \ref{thmquadratic}),
\begin{eqnarray*}
((A\alpha^2+B\alpha+C)^2)_q&\equiv&-(B^2-4AC)_q\mod q,\text{ or}\\
((D\alpha^2+E\alpha+F)\cdot(H\alpha+J)^2)_q&\equiv&-(DJ^2-EHJ+FH^2)_q\mod q,
\end{eqnarray*}
if the total degrees of the polynomials on the left hand sides do not exceed twice the number of their variables.
The situation where one has both cases simultaneously implies a factorization $(W\alpha+X)^2(Y\alpha+Z)^2$. In this case both reductions are equivalent
to standard denominator reduction. As examples, $(W\alpha+X)^3(Y\alpha+Z)$ reduces by the second case to zero whereas $(U\alpha^2+V\alpha+W)(X\alpha^2+Y\alpha+Z)$
or $(U\alpha+V)(W\alpha+X)(Y\alpha+Z)$ do not reduce in general.

We found that quadratic denominator reduction is surprisingly efficient. An instructive example for the potential of the technique are reductions of infinite families of
graphs with hourglass chains in \cite{Hourglass}.

In Section \ref{sectini} we prove that one can always do a minimum of nine reductions, in most cases much more.
We use the power of the new method to examine all $c_2$s of loop order 11 in $\phi^4$ and 10 in non-$\phi^4$. In $\phi^4$ at loop orders 12 and 13 we can still perform
a partial analysis which contains thousands of graphs. Altogether we can discriminate a total of 4801 $c_2$s. The main result is that the sparsity found in \cite{BSmod} is confirmed.
A possible counter-example of $\leq11$ loops to Conjecture \ref{nowt2} (absence of weight 2) has to be beyond level 2000 in the modular form, see Result \ref{result}.

Explicitly, our findings are summarized in Table \ref{tab1} where we introduced a notion of dimension for $c_2$s (see Definition \ref{defdim}).
In case of a modular correspondence the dimension of the $c_2$ is the weight of the modular form minus 1. The correspondence is proved in \cite{K3} for the
weight 3 level 7 of dimension 2. More proved results are in \cite{Lproofs} for the [weight, level]  modular forms [3, 8], [3, 12], [4, 13], [6, 7]. All data in \cite{BSmod} and
in this article are conjectural and proved only up to the primes indicated as superscripts in Table \ref{tab1}.

Note that most $c_2$s are not in Table \ref{tab1} because they could not be identified as constants, Legendre symbols, or Fourier coefficients of modular forms.
Concretely, we conjecture that unidentified sequences begin at $c_2$-dimension 4 in $\phi^4$ and $c_2$-dimension 3 in non-$\phi^4$ (see Conjecture \ref{condim23}).

In Sections \ref{sectdp} to \ref{sectdr} we basically review known results for Dodgson polynomials, the $c_2$, and standard denominator reduction.
To prepare the upcoming new material, we develop a new approach to signs for Dodgson polynomials in Section \ref{sectdp}.

In Sections \ref{sectls} to \ref{sectpf} we define and prove quadratic denominator reduction.

Section \ref{sectini} provides results for initial reductions.

In Sections \ref{sectaffred} and \ref{sectres} we focus on the calculation of prefixes (initial prime sequences) of $c_2$s up to loop order 13.

We conclude the main part of the article in Section \ref{sectconj} with a (somewhat personal) account on conjectures, speculations, and open problems related to the $c_2$.
The main interest is the identification of perturbative geometries. Regretfully, we do not even conjecturally have a description of perturbative geometries in general.
Still, we speculate that $\phi^4$ geometries have (at least) two ingredients. One is an analytical property: We expect that any perturbative geometry ($\phi^4$ or not) is the intersection
of two minimal polynomials of Dodgson type, see Definition \ref{defdodpair} and Conjecture \ref{condodpair}.
For dimension zero this restricts to quadratic extensions of $\FF_q$ with discriminant $\not\equiv3\mod4$ while in the one-dimensional case only (some) elliptic curves are allowed.
The second (unknown) property seems to be of arithmetic type: At zero dimensions it relates to cylcotomic extensions and at dimension one it should rule out {\em all} curves.
The arithmetic property only seems to exist in $\phi^4$ which genuinely makes it a quantum condition.
In the author's personal view, specifying and understanding this arithmetic property is the main problem in the field, see Problem \ref{probgeo}.

In the appendix Friedrich Knop proves an extension of the Chevalley-Warning-Ax theorem \cite{Ax,Sa} to double covers of affine space. The appendix is of different style as it
uses deep results in algebraic geometry. The result of the appendix is vital to fully lift the results of the main sections from primes to prime powers.

All computations were performed using \cite{Shlog} on spare office PCs at the Department Mathematik, University Erlangen-N\"urnberg, Germany.

\section*{Acknowledgements}
This work was driven by the ongoing interest of K. Yeats into the $c_2$-invariant. Her student J. Shaw first took on the challenge of 11 loop calculations with some
very valuable results. Many thanks also to Matthias Bauer and Martin Bayer who provided all the necessary computing resources. The author is grateful to the referees who
made many very valuable suggestions. He is also indebted to F. Knop for providing the appendix. The author is supported by the DFG grant SCHN~1240.

\setlength{\tabcolsep}{2pt}
\renewcommand{\arraystretch}{1.45}

\begin{table}
\begin{center}
\begin{tabular}{l|llllllllll}
0\hspace*{7mm}
 &1\hspace*{7mm}
  &2\hspace*{7mm}
   &3\hspace*{7mm}
    &\hspace*{7mm}
     &4\hspace*{7mm}
      &5\hspace*{7mm}
       &6\hspace*{7mm}
        &7\hspace*{7mm}
         &8\hspace*{7mm}
          &9\\
$(\frac{x}{q})$&\multicolumn{10}{l}{level $x$ modular form of weight $=$ dimension $+$ 1}\\\hline
$\setlength{\fboxrule}{0.5mm}\fbox{0}_{^{\scriptstyle 6}}$
 &$\setlength{\fboxrule}{0.1mm}\fbox{11}_{^{\scriptstyle 9}}^{_{\scriptstyle 97}}$
  &$\setlength{\fboxrule}{0.5mm}\fbox{7}_{^{\scriptstyle 8}}$
   &$\setlength{\fboxrule}{0.5mm}\fbox{5}_{^{\scriptstyle 8}}^{_{\scriptstyle 97}}$
    &&$\setlength{\fboxrule}{0.5mm}\fbox{4}_{^{\scriptstyle 9}}^{_{\scriptstyle 97}}$
      &$\setlength{\fboxrule}{0.5mm}\fbox{3}_{^{\scriptstyle 8}}^{_{\scriptstyle 97}}$
       &$\setlength{\fboxrule}{0.5mm}\fbox{3}_{^{\scriptstyle 9}}^{_{\scriptstyle 97}}$
        &$\setlength{\fboxrule}{0.5mm}\fbox{2}_{^{\scriptstyle 10}}^{_{\scriptstyle 47}}$
         &$\setlength{\fboxrule}{0.5mm}\fbox{4}_{^{\scriptstyle\leq23}}^{_{\scriptstyle 41}}$
          &2\\
$\setlength{\fboxrule}{0.5mm}\fbox{1}_{^{\scriptstyle 3}}$
 &$\setlength{\fboxrule}{0.1mm}\fbox{14}_{^{\scriptstyle 9}}^{_{\scriptstyle 97}}$
  &$\setlength{\fboxrule}{0.5mm}\fbox{8}_{^{\scriptstyle 8}}$
   &$\setlength{\fboxrule}{0.5mm}\fbox{6}_{^{\scriptstyle 9}}^{_{\scriptstyle 97}}$
    &$\setlength{\fboxrule}{0.1mm}\fbox{19}_{^{\scriptstyle 10}}^{_{\scriptstyle 97}}$
     &7
      &$\setlength{\fboxrule}{0.5mm}\fbox{4}_{^{\scriptstyle 9}}^{_{\scriptstyle 97}}$
       &7
        &3
         &7
          &$\setlength{\fboxrule}{0.5mm}\fbox{3}_{^{\scriptstyle 11}}^{_{\scriptstyle 41}}$\\
$-1$
 &$\setlength{\fboxrule}{0.1mm}\fbox{15}_{^{\scriptstyle 9}}^{_{\scriptstyle 97}}$
  &$\setlength{\fboxrule}{0.1mm}\fbox{11}_{^{\scriptstyle 9}}^{_{\scriptstyle 97}}$
   &$\setlength{\fboxrule}{0.5mm}\fbox{7}_{^{\scriptstyle 10}}^{_{\scriptstyle 47}}$
    &20
     &$\setlength{\fboxrule}{0.5mm}\fbox{8}_{^{\scriptstyle 12}}^{_{\scriptstyle 97}}$
      &5
       &8
        &$\setlength{\fboxrule}{0.5mm}\fbox{5}_{^{\scriptstyle 10}}^{_{\scriptstyle 97}}$
         &8
          &3\\
$\setlength{\fboxrule}{0.1mm}\fbox{2}_{^{\scriptstyle 10}}^{_{\scriptstyle 47}}$
 &$\setlength{\fboxrule}{0.1mm}\fbox{17}_{^{\scriptstyle 10}}^{_{\scriptstyle 97}}$
  &$\setlength{\fboxrule}{0.5mm}\fbox{12}_{^{\scriptstyle 9}}$
   &$\setlength{\fboxrule}{0.5mm}\fbox{8}_{^{\scriptstyle 11}}^{_{\scriptstyle 97}}$
    &$\setlength{\fboxrule}{0.5mm}\fbox{21}_{^{\scriptstyle 11}}^{_{\scriptstyle 97}}$
     &11
      &6
       &11
        &6
         &11
          &4\\
$-2$
 &19
  &$\setlength{\fboxrule}{0.1mm}\fbox{15}_{^{\scriptstyle 10}}^{_{\scriptstyle 97}}$
   &9
    &21
     &12
      &$\setlength{\fboxrule}{0.5mm}\fbox{7}_{^{\scriptstyle 9}}$
       &15
        &7
         &12
          &5\\
3
 &20
  &15
   &$\setlength{\fboxrule}{0.5mm}\fbox{10}_{^{\scriptstyle 11}}^{_{\scriptstyle 97}}$
    &$\setlength{\fboxrule}{0.1mm}\fbox{22}_{^{\scriptstyle 10}}^{_{\scriptstyle 97}}$
     &15
      &8
       &15
        &8
         &15
          &6\\
$\setlength{\fboxrule}{0.5mm}\fbox{-3}_{^{\scriptstyle 7}}$
 &21
  &16
   &$\setlength{\fboxrule}{0.5mm}\fbox{12}_{^{\scriptstyle 12}}^{_{\scriptstyle 97}}$
    &22
     &15
      &9
       &16
        &8
         &15
          &8\\
$\setlength{\fboxrule}{0.5mm}\fbox{4}_{^{\scriptstyle 7}}$
 &24
  &19
   &$\setlength{\fboxrule}{0.5mm}\fbox{13}_{^{\scriptstyle 9}}$
    &22
     &19
      &$\setlength{\fboxrule}{0.5mm}\fbox{10}_{^{\scriptstyle 10}}^{_{\scriptstyle 47}}$
       &19
        &9
         &19
          &8\\
$\setlength{\fboxrule}{0.5mm}\fbox{-4}_{^{\scriptstyle 8}}$
 &26
  &$\setlength{\fboxrule}{0.1mm}\fbox{20}_{^{\scriptstyle 10}}^{_{\scriptstyle 97}}$
   &14
    &$\setlength{\fboxrule}{0.1mm}\fbox{23}_{^{\scriptstyle 10}}^{_{\scriptstyle 97}}$
     &20
      &10
       &20
        &10
         &20
          &9\\
$\setlength{\fboxrule}{0.1mm}\fbox{5}_{^{\scriptstyle 9}}^{_{\scriptstyle 97}}$
 &26
  &20
   &14
    &$\vdots$
     &20
      &10
       &20
        &12
         &20
          &9\\
$-5$
 &27
  &$\setlength{\fboxrule}{0.5mm}\fbox{24}_{^{\scriptstyle 11}}^{_{\scriptstyle 97}}$
   &15
    &$\setlength{\fboxrule}{0.1mm}\fbox{32}_{^{\scriptstyle 10}}^{_{\scriptstyle 97}}$
     &24
      &$\setlength{\fboxrule}{0.5mm}\fbox{11}_{^{\scriptstyle 11}}^{_{\scriptstyle 47}}$
       &23
        &12
         &24
          &9\\
$\vdots$
 &$\vdots$
  &$\vdots$
   &15
    &$\vdots$
     &$\vdots$
      &$\vdots$
       &$\vdots$
        &$\vdots$
         &$\vdots$
          &$\vdots$\\
$\setlength{\fboxrule}{0.5mm}\fbox{9}_{^{\scriptstyle 13}}$
 &$\setlength{\fboxrule}{0.1mm}\fbox{36}_{^{\scriptstyle 10}}^{_{\scriptstyle 97}}$
  &&$\setlength{\fboxrule}{0.5mm}\fbox{16}_{^{\scriptstyle 12}}^{_{\scriptstyle 31}}$
    &$\setlength{\fboxrule}{0.1mm}\fbox{56}_{^{\scriptstyle 10}}^{_{\scriptstyle 47}}$
     &&$\setlength{\fboxrule}{0.1mm}\fbox{17}_{^{\scriptstyle 10}}^{_{\scriptstyle 47}}$
       &&&&\\
$\vdots$
 &$\setlength{\fboxrule}{0.1mm}\fbox{37}_{^{\scriptstyle 10}}^{_{\scriptstyle 97}}$
  &&$\setlength{\fboxrule}{0.5mm}\fbox{17}_{^{\scriptstyle 10}}^{_{\scriptstyle 97}}$
    &$\vdots$
     &&$\vdots$
       &&&&\\
$\setlength{\fboxrule}{0.5mm}\fbox{-12}_{^{\scriptstyle 12}}$
 &$\vdots$
  &&18
    &$\setlength{\fboxrule}{0.5mm}\fbox{73}_{^{\scriptstyle 11}}^{_{\scriptstyle 97}}$
     &&&&&&\\
$\vdots$
 &&&&$\vdots$
     &&&&&&
\end{tabular}
\end{center}
\caption{\small Identified $c_2$-invariants of various dimensions (first row). At dimension zero $-c_2$ is a Legendre symbol $(x/q)\mod q$ for various values of $x$, see Definition \ref{deflege}.
For dimensions $\geq1$ the $c_2$-invariants are identified via modularity with respect to newforms with integer Fourier coefficients (see Definition \ref{defmod}).
The weight of the newform is the dimension of the $c_2$ plus 1. In the columns the levels $x$ of identified newforms are boxed.
The lower index of an identified $c_2$-invariant gives the loop order of its first appearance in $\phi^4$ theory (bold box) or in non-$\phi^4$ theory (thin box).
The upper index indicates the maximum prime up to which the identification has been confirmed (not necessarily at the loop order of its first appearance). Identifications with no superscript
have been proved for all primes, either here (at dimension zero for $q=2$ and odd prime powers) or in \cite{K3,Lproofs}.
The newform of level 4 at dimension 8 was identified in \cite{Hourglass} with an upper bound $\leq23$ for the loop order.
The table is possibly incomplete for $\geq12$ loops in $\phi^4$. It does not include any non-$\phi^4$ $c_2$s of graphs with $\geq11$ loops.}
\label{tab1}
\end{table}

\section{Dodgson polynomials}\label{sectdp}
Calculating the $c_2$-invariant relies on point-counting Dodgson polynomials over finite fields $\FF_q$.
Because at high loop orders these polynomials are big, this task seems not efficient for large primes. However, the (very special) Dodgson polynomials fulfill a
plethora of identities which one can utilize to reduce the point-count to smaller polynomials.

In previous work \cite{K3,BSYc2,Dc2,Yc2,Yc2pre}, even at the last reduction step, it was sufficient to count the zeros of polynomials.
The overall signs did not matter and relations were often only derived up to sign. More recently \cite{HSSYc2}, point-counts of sums of Dodgson polynomials were considered,
so that the signs of the individual terms mattered. Here, we use a refined reduction where even for single polynomials the overall sign is important.
So, all necessary relations are re-derived with a revised sign convention for Dodgson polynomials. We do this solely in the framework of determinant relations.

The setup is the following: Let $G$ be a (multi-)graph which may have several components, multiple edges, and self-loops (i.e.\ edges which begin and end in the same vertex).
The graph $G$ has the edge set $\sE(G)$ and the vertex set $\sV(G)$. We pick an arbitrary orientation on the edges of $G$ and write edges as two letter words of vertices.
So, $uv$ is an edge that begins in $u$ and ends in $v$. A self-loop is $vv$. Note that in $G$ may exist several edges $uv$ or $vv$.

We order the edges and the vertices of $G$ in some arbitrary way
$$
\iota:\sE(G)\cup\sV(G)\rightarrow\{1,\ldots,|\sE(G)|+|\sV(G)|\}
$$
so that $\iota_i$ is the position of the edge or vertex $i$ in the following (full, symmetric) expanded Laplacian:
$$
L(G)_{i,j} = \left\{\begin{array}{rl}
\alpha_e, &\hbox{if }i=j=\iota_e\hbox{ for } e\in\sE(G),\\
       1, &\hbox{if }\{i,j\}=\{\iota_{uv},\iota_u\}\hbox{ for } u\neq v\in\sV(G),\\
      -1, &\hbox{if }\{i,j\}=\{\iota_{vu},\iota_u\}\hbox{ for } u\neq v\in\sV(G),\\
       0, &\hbox{otherwise}.
\end{array}\right.
$$
To define $L(G)$ we introduced variables $\alpha_e$ for every edge $e$ in $G$. In most equations we suppress the dependence on the variables $\alpha$.
The symmetric matrix $L(G)$ depends on the orientation of the edges and the choice of $\iota$. In the following we use edge and vertex labels ($e$, $u$ rather than the position $\iota$)
to refer to rows and columns in $L(G)$.

If one orders edges before vertices, $L(G)$ can be expressed in terms of a diagonal matrix $A$ carrying the variables $\alpha$ and the signed incidence matrix $E_G$.
\begin{equation}\label{Lincidence}
L(G)=\left(\begin{array}{c|c}
  A & E_G^T\\\hline
E_G & 0\end{array}\right)\quad\text{with}\quad (E_G)_{v,e}=\left\{\begin{array}{rl}
       1, &\hbox{if }e=vu,\;u\in\sV(G),\\
      -1, &\hbox{if }e=uv,\;u\in\sV(G),\\
       0, &\hbox{otherwise}.
\end{array}\right.
\end{equation}

\begin{ex}\label{ex1}
For a pair of self-loops (A), a double edge (B), a chain of two edges (C) and two disconnected edges (D) we obtain the following expanded graph Laplacians (for suitable
orientations and orders of edges and vertices).
\begin{eqnarray*}
L(A)=\left(\begin{array}{ccc}
 \alpha_1&0&0\\
 0&\alpha_2&0\\
 0&0&0\end{array}\right),\hspace{10mm}&&
L(B)=\left(\begin{array}{cccc}
 \alpha_1&0&1&-1\\
 0&\alpha_2&1&-1\\
 1&1&0&0\\
 -1&-1&0&0\end{array}\right),\\
L(C)\!=\!\left(\begin{array}{ccccc}
 \alpha_1&0&1&-1&0\\
 0&\alpha_2&0&1&-1\\
 1&0&0&0&0\\
 -1&1&0&0&0\\
 0&-1&0&0&0\end{array}\right)\!,&&
\hspace{-5mm}L(D)\!=\!\left(\begin{array}{cccccc}
 \alpha_1&0&1&-1&0&0\\
 0&\alpha_2&0&0&1&-1\\
 1&0&0&0&0&0\\
 -1&0&0&0&0&0\\
 0&1&0&0&0&0\\
 0&-1&0&0&0&0\end{array}\right)\!.
\end{eqnarray*}
\end{ex}
For $A,B\subseteq\sE(G)\cup\sV(G)$ and $K\subseteq\sE(G)$ we define the reduced matrix
$$
L^{A,B}_K(G)=\left.L^{A,B}(G)\right|_{\alpha_k=0,\,k\in K}\,,
$$
where $L^{A,B}$ is $L$ with rows $A$ and columns $B$ removed.

For our sign convention we lift the upper indices to words, i.e.\ we keep track of the order of elements (letters) in $A$ and $B$. In case of repeated letters we will define
that Dodgsons vanish. For any words we define a sign relative to $\iota$.

\begin{defn}
Let $A=a_1a_2\cdots a_k$ be a word in letters $a_i$ ordered by $\iota$. Then
$$
\sgn(A)=\left\{\begin{array}{cl}
        0 &\hbox{, if $A$ has repeated letters},\\
\sgn(\pi) &\hbox{, if $\iota_{a_{\pi(1)}}<\iota_{a_{\pi(2)}}<\cdots<\iota_{a_{\pi(k)}}$}.
\end{array}\right.
$$
\end{defn}

\begin{ex}
If 1,2 are in natural order then $\sgn(\emptyset)=\sgn(1)=\sgn(2)=\sgn(12)=1$, $\sgn(21)=-1$, $\sgn(11)=\sgn(22)=0$.
\end{ex}
If we define
$$
A_{<x}=\{a\in A:\iota_a<\iota_x\},
$$
we obtain
\begin{equation}\label{Ax}
\sgn(Ax)=(-1)^{|A|-|A_{<x}|}\sgn(A),\quad\hbox{if }x\notin A.
\end{equation}
Permutations of $A$ act on both sides of this equation with the same sign, so we may assume that $A$ is ordered according to $\iota$. To order $Ax$ we need $|A|-|A_{<x}|$ transpositions.

\begin{remark}\label{setrem}
It is convenient to extend the set theory symbols $\in,\cap,\cup,\backslash$ to words. Whenever we use these symbols on a word we consider the word as the set of its letters.
For any words $A$ and $B$ we e.g.\ write $x\in A$ if $x$ is a letter in $A$ or $A\cap B$ for the set of common letters in $A$ and $B$. In the context of Dodgsons $\Psi^{I,J}_K$
(defined below) where $I$ and $J$ are words and $K$ is a set we opt for the most compact notation and frequently concatenate words before using set symbols.
We e.g.\ write $IJ\cup K$ or $K\backslash IJ$ instead of $I\cup J\cup K$ or $K\backslash(I\cup J)$, respectively. We also use the notation
$$
\iota_A=\sum_{x\in A}\iota_x
$$
for any word (or set) A.
\end{remark}

\begin{defn}\label{defpsi}
Let $G$ be a graph and $v\in\sV(G)$. Let $I,J$ be words of equal length in the letters $\sE(G)$ and let $K\subseteq\sE(G)$. Then the {\em Dodgson} is defined as
\begin{equation}\label{psidef}
\Psi^{I,J}_K(G)=(-1)^{|\sV(G)|+\iota_I+\iota_J-1}\sgn(Iv)\,\sgn(Jv)\det L^{Iv,Jv}_K(G).
\end{equation}
The determinant of the empty matrix is 1. For convenience we define $\Psi^{I,J}_K(G)=0$ if $|I|\neq|J|$.
We omit empty indices; the Dodgson $\Psi(G)$ is the graph polynomial (\ref{1}).
\end{defn}
If we need explicit reference to the variables $\alpha$ we may drop the argument $G$ or write it as subscript.
In particular, we write $\Psi(\alpha)$ or $\Psi_G(\alpha)$ for the graph polynomial. For the connection to (\ref{1}) see \cite{Bperiods}.

\begin{ex}\label{ex2}
Continuing Example \ref{ex1} we have the following graph polynomials and Dodgsons (for any choice of $\iota$ and $v$)
$$
\begin{array}{llll}
\Psi(A)=\alpha_1\alpha_2,&\quad\Psi(B)=\alpha_1+\alpha_2,&\quad\Psi(C)=1,&\quad\Psi(D)=0,\\
\Psi^{1,1}(A)=\alpha_2,&\quad\Psi^{1,1}(B)=1,&\quad\Psi^{1,1}(C)=0,&\quad\Psi^{1,1}(D)=0,\\
\Psi^{1,2}(A)=0,&\quad\Psi^{1,2}(B)=-1,&\quad\Psi^{1,2}(C)=0,&\quad\Psi^{1,2}(D)=0.
\end{array}
$$
\end{ex}

\begin{lem}
The Dodgson polynomial $\Psi^{I,J}_K(G)$ does not depend on $\iota$ or the choice of $v$.
\end{lem}
\begin{proof}
To prove independence of $\iota$ we may assume that neither $I$ nor $J$ have repeated letters.
Let $M=(m_{i,j})$ be a generic square matrix labeled by $\sE(G)\cup\sV(G)$. Let $A,B$ be words of equal length with distinct letters in $\sE(G)\cup\sV(G)$.
We prove by induction over $|A|=|B|$ that
\begin{equation}\label{psiAB}
\Phi^{A,B}:=(-1)^{\iota_A+\iota_B}\sgn(A)\,\sgn(B)\det M^{A,B}
\end{equation}
is independent of $\iota$.

If $A$ and $B$ are empty then a change of $\iota$ amounts to conjugating $M$ by a permutation matrix. This leaves the determinant invariant.

Now, consider the words $Ax,By$ for letters $x\notin A,y\notin B$. The coefficient of $m_{x,y}$ in $\det M^{A,B}$ is $(-1)^{\iota_x-|A_{<x}|+\iota_y-|B_{<y}|}\det M^{Ax,By}$.
We multiply the above equation with the sign on the right hand side of (\ref{psiAB}) and use induction. The result follows from (\ref{Ax}) because $|A|=|B|$.
With $A=Iv$, $B=Jv$, $M=L_K$ in (\ref{psiAB}) we see that $\Psi^{I,J}_K$ does not depend on $\iota$.

To show the independence of $v$ let $u\neq v$ be a another vertex in $G$. By independence of $\iota$ we may assume that edges are ordered before vertices and use
(\ref{Lincidence}) for the expanded Laplacian. Let $r_u$, $r_v$ be rows $u$, $v$ in $E_G$, respectively. Because the rows add up to zero in $E_G$,
adding the rows (columns) $\neq u$ in the lower left (upper right) corner of $L^{Iv,Jv}_K$ to $r_u$ ($r_u^T$) gives $-r_v$ ($-r_v^T$).
A simultaneous multiplication of row and column $u$ by $-1$ (leaving the determinant unchanged) gives $L^{Iu,Ju}_K$ up to a transposition of the vertices $u$ and $v$.
Independence of $\iota$ proves the claim.
\end{proof}

In the following we will frequently reduce graphs to minors by deleting and contracting edges.
We write $G\backslash e$ if we remove the edge $e$ from $G$ and $G/e$ if we contract the edge $e$ in $G$, i.e.\ we remove $e$ and identify the end-points of $e$.
Contraction leaves the number of components $h_0(G)$ and the number of independent cycles $h_1(G)$ unchanged. Note that contraction and deletion commute.

Only edges which are not self-loops can be contracted. This makes it convenient to pass to the free abelian group $\ZZ[\sM(G)]$ generated by the set $\sM(G)$ of minors (possibly with self-loops)
of a graph $G$. This notation implies that we use the natural identification of edges in minors with edges in $G$ to
establish a coherent set of variables in $\sM(G)$. In some contexts it is also useful to make $\ZZ[\sM(G)]$ a ring where multiplication is disjoint union.
In $\ZZ[\sM(G)]$ we identify a graph $H$ with $1\cdot H$ and define
$$
H/e=0,\quad\text{if $e$ is a self-loop in $H\in\sM(G)$.}
$$
Moreover, we extend the definition of the Dodgson $\Psi^{I,J}_K$ to $\ZZ[\sM(G)]$ by linearity. In particular, we have $\Psi^{I,J}_K(0)=0$.
\begin{remark}
Using set notations for words as explained in Remark \ref{setrem} we get the following elementary results.
\begin{enumerate}
\item $\Psi^{I,J}_K=\Psi^{J,I}_K$ by transposition.
\item The orientation of an edge $e$ affects $\Psi^{I,J}_K$ by a sign if and only if $e\in I\backslash J$ or $e\in J\backslash I$, i.e.\ $e$ is in exactly one of the two words $I$, $J$.
\item The Dodgson $\Psi^{I,J}_K$ is of degree $\leq1$ in all variables. If $e\in IJ\cup K$ then $\Psi^{I,J}_K$ is constant in $\alpha_e$.
In this case nullifying $\alpha_e$ is trivial:
\begin{equation}\label{red1}
\Psi^{I,J}_K=\Psi^{I,J}_{K\backslash IJ}.
\end{equation}
\end{enumerate}
\end{remark}

A key connection between the Dodgson polynomial and the topology of the underlying graph is given by the following contraction-deletion lemma:
\begin{lem}[contraction-deletion]\label{lemcd}
Let $I,J$ be words in $\sE(G)$ and $K\subseteq\sE(G)$. For any $e\in\sE(G)\backslash (IJ\cup K)$,
\begin{equation}\label{cd}
\Psi^{I,J}_K(G)=\alpha_e\Psi^{Ie,Je}_K(G)+\Psi^{I,J}_{Ke}(G)=\alpha_e\Psi^{I,J}_K(G\backslash e)+\Psi^{I,J}_K(G/e).
\end{equation}
\end{lem}
\begin{proof}
We order edges before vertices with $\iota_e=1$. Omitting $e$ induces the ordering $\iota'=\iota-1$ on $G\backslash e$. By definition $L^{Iev,Jev}_K(G)=L^{Iv,Jv}_K(G\backslash e)$.
Moreover, $\iota_{Ie}=\iota'_I+|I|+1$ and, by (\ref{Ax}), $\sgn(Iev)=\sgn(Ie)=(-1)^{|I|}\sgn(I)$. With the analogous statements for $Je$ we get $\Psi^{Ie,Je}_K(G)=\Psi^{I,J}_K(G\backslash e)$.

Next, we show that $\Psi^{I,J}_{Ke}(G)=\Psi^{I,J}_K(G/e)$. If $e$ is a self-loop then row and column $e$ are zero in $L^{Iv,Jv}_{Ke}$ and the result follows.
Otherwise orient $e$ such that $e=uv$ and assume $\iota_u=|\sE(G)|+|\sV(G)|$. Because $\alpha_e=0$ in $L^{Iv,Jv}_{Ke}(G)$, row and column 1 have non-zero entries only at positions $\iota_u-|J|-1$
and $\iota_u-|I|-1$, respectively. Both entries are 1. Expanding the determinant along row and (thereafter) column $1$ gives the sign $(-1)^{\iota_u-|J|}(-1)^{\iota_u-|I|-1}$.
The deletion of row and column $1$ can be interpreted as deleting edge $e$ and vertex $u$ while connecting all $u$-incident edges to $v$ (keeping the orientation).
This gives $G/e$. We conclude that $\det L^{Iv,Jv}_{Ke}(G)=(-1)^{|I|+|J|+1}\det L^{Iv,Jv}_{K}(G/e)$.
Contracting $e$ induces the ordering $\iota'=\iota-1$ on $G/e$. We get $\iota_I+\iota_J=\iota'_I+\iota'_J+|I|+|J|$.
With a minus sign from the reduction of the number of vertices in $G/e$ we get $\Psi^{I,J}_{Ke}(G)=\Psi^{I,J}_K(G/e)$.

Finally, we show that the coefficient of $\alpha_e$ in $\Psi^{I,J}_K(G)$ is $\Psi^{Ie,Je}_K(G)$.
The only entry with $\alpha_e$ in $L^{Iv,Jv}_K(G)$ is at position $(1,1)$. The coefficient of $\alpha_e$ is hence $\det L^{Iev,Jev}_K(G)$. The result follows from (\ref{Ax}).
\end{proof}

\begin{remark}
By passing to minors, iterated use of (\ref{cd}) [together with (\ref{red1})] allows one to reduce Dodgson polynomials to the case $I\cap J=K=\emptyset$,
$$
\Psi^{I,J}_K(G)=\Psi^{I\backslash J,J\backslash I}(G\backslash(I\cap J)/(K\backslash IJ)).
$$
\end{remark}

A cut set is a set of edges that, when removed, leaves a disconnected graph. Cuts in $I$ (or in $J$) or cycles in $K$ trivialize Dodgsons.
\begin{lem}[cuts and loops, F. Brown, \S 2.2 (4) in \cite{BSYc2}]\label{lemcuts}
If $I$ cuts $G$ or $JK\backslash I$ contains a cycle (i.e.\ $h_0(G\backslash I)\geq2$ or $h_1(JK\backslash I)\geq1$) then $\Psi^{I,J}_K(G)=0$.
In particular, any disconnected graph has trivial Dodgsons.
\end{lem}
\begin{proof}
In $L^{Iv,Jv}_K$ the rows of edges in $I$ and the row of the vertex $v$ are deleted. The columns in $L^{Iv,Jv}_K$ corresponding to the vertices of the connected component
not containing $v$ add up to zero. Hence $\Psi^{I,J}_K=0$.

Assume the cycle $C$ is in $JK\backslash I$. Using (\ref{red1}) we assume that $K\cap IJ=\emptyset$.
If an edge $e\in C$ is in $K$ we contract $e$ using (\ref{cd}). This process iterates until either a self-loop which contracts to zero or $C\cap K=\emptyset$.
We can hence assume that $C$ (which may be a self-loop) is in $J\backslash I$. The rows corresponding to $C$ in $L^{Iv,Jv}_K$ are of the form $(0,E_C^T)$.
Because $C$ is a cycle they add up to zero and $\Psi^{I,J}_K=0$.
\end{proof}

We need the notion of oriented vertices and cycles.
\begin{defn}
We call $v\in\sV(G)$ {\em oriented} if non-self-loop edges incident to $v$ are all outgoing or all ingoing.
We write $e\sim v$ if the non-self-loop edge $e$ is incident to the oriented vertex $v$ (which may have incident self-loops $\neq e$).
A cycle $C$ is {\em oriented} if all edge-heads are attached to tails in $C$.
\end{defn}

\begin{lem}[vanishing sums]
Let $G$ be a graph, $I,J$ words in the letters $\sE(G)$, and $K\subseteq\sE(G)$. If $v\in\sV(G)$ is oriented and $e\sim v$, $e\notin IJ$, then
\begin{equation}\label{v0}
\sum_{f\sim v}\Psi^{If,J}_{Ke}(G)=0.
\end{equation}
If $C$ is an oriented cycle attached to $G$, i.e.\ $\sE(C)\cap\sE(G)=\emptyset$ and $\sV(C)\subseteq\sV(G)$, then
\begin{equation}\label{c0}
\sum_{f\in C}\Psi^{If,J}_K(G\cup\{f\})=0.
\end{equation}
\end{lem}
\begin{proof}
Note that terms in (\ref{v0}) with $f\in I$ drop because they have a repeated letter in a superscript.
Using (\ref{red1}) and (\ref{cd}) we write (\ref{v0}) in the form
\begin{equation}\label{v0a}
\Psi^{Ie,J}_K(G)=-\sum_{Ie\not\ni f\sim v}\Psi^{If,J}_K(G/e).
\end{equation}
We take the second vertex $u$ of $e$ as the vertex whose row and column is deleted in the graph Laplacian for calculating the Dodgson.
We order edges before vertices with $\iota_e=1$ and $\iota_v=|\sE(G)|+|\sV(G)|$.
Then, the only non-zero entries in columns 1 and $v$ in $L^{Ieu,Ju}_K$ are at positions $(\iota_v-|Ie|-1,1)$ and $(\iota_f-|Ie_{<f}|,\iota_v-|J|-1)$ for $Ie\not\ni f\sim v$, respectively.
The entries are $\sigma=+1$ if $e$ (and hence also $f$) begins in $v$ and $\sigma=-1$ otherwise. We first expand along column $1$ and thereafter along column $v$ yielding
\begin{eqnarray*}
\det L^{Ieu,Ju}_K&=&\sigma(-1)^{\iota_v-|Ie|}\det L^{Ievu,Jeu}_K,\\
\det L^{Ievu,Jeu}_K&=&\sum_{Ie\not\ni f\sim v}\sigma(-1)^{\iota_f-|Ie_{<f}|+\iota_v-|Je|-1}\det L^{Iefvu,Jevu}_K,\\
L^{Iefvu,Jevu}_K(G)&=&L^{Ifu,Ju}_K(G/e),
\end{eqnarray*}
with the ordering $\iota'=\iota-1$ on $G/e$. For any $f\notin Ie$ we get $\iota_{Ie}=\iota'_{I}+|Ie|=\iota'_{If}-\iota_f+|Ie|$ and $\iota_J=\iota'_J+|J|$.
Collecting all signs we obtain (the $\sgn$'s with respect to $\iota$ and $\iota'$ are identical)
$$
\Psi^{Ie,J}_K(G)=\sgn(Ie)\sum_{Ie\not\ni f\sim v}(-1)^{-|Ie_{<f}|}\sgn(If)\Psi^{If,J}_K(G/e),
$$
With (\ref{Ax}) we get $\sgn(If)=(-1)^{|I|-|I_{<f}|}\sgn(I)=(-1)^{-|I_{<f}|}\sgn(Ie)$ and (\ref{v0a}) follows because $|I_{<f}|=|Ie_{<f}|-1$.

For (\ref{c0}) we order the graph $G\cup\{f\}$ such that $\iota_f=1$. The matrices $L^{Ifu,Ju}_K(G\cup\{f\})$ are identical after the first column.
Because $C$ is an oriented cycle the first columns of the matrices $L^{Ifu,Ju}_K(G\cup\{f\})$ add up to zero yielding $\sum_{f\in C}\det L^{Ifu,Ju}_K(G\cup\{f\})=0$.
Let $\iota'=\iota-1$ be the ordering on $G$, then $\iota_{If}=\iota'_I+|If|$ and $\iota_J=\iota'_J+|J|$. For $\Psi^{If,J}_K(G\cup\{f\})$ we get the expression
$$
(-1)^{|\sV(G\cup\{f\})|+\iota'_I+|I|+\iota'_J+|J|}\sgn(If)\,\sgn(J)\det L^{Ifu,Ju}_K(G\cup\{f\}).
$$
We have $|\sV(G\cup\{f\})|=|\sV(G)|$ and by (\ref{Ax}) the sign becomes independent of $f$. The expression sums up to zero.
\end{proof}

The most practical way to use (\ref{v0}) is by (\ref{v0a}). The result simplifies in a special setup with a 2- or 3-valent vertex.

\begin{lem}[see Proposition 1.16 in \cite{HSSYc2}]
Let $G$ be graph and $u\in\sV(G)$ an oriented vertex with two edges $e,f$ or three edges $e,f,h$ and no self-loops.
Let $I,J$ be words in $\sE(G)$ and $e,f\notin IJ$. If $u$ has degree three then assume furthermore that $h\in J\backslash I$. For any $K\subseteq\sE(G)$ we have
\begin{equation}\label{23u}
\Psi^{Ie,Jf}_K(G)=-\Psi^{I,J}_K(G\backslash e/f).
\end{equation}
\end{lem}
\begin{proof}
Consider the case that $u$ has degree 3. We may permute edge $h$ to the rightmost position of $J$. Because this induces the same sign on both sides of (\ref{23u}) we may assume that $J=J'h$.
We use (\ref{v0a}) and (\ref{cd}) to obtain
\begin{eqnarray*}
\Psi^{Ie,Jf}_K(G)&=&-\;\Psi^{If,J'hf}_K(G/e)-\Psi^{Ih,J'hf}_K(G/e)\\
&=&-\;\Psi^{I,J'h}_K(G\backslash f/e)+\Psi^{I,J'f}_K(G\backslash h/e).
\end{eqnarray*}
For the last identity we swapped $h$ with $f$ which gives a minus sign.

If $u$ has degree 2 then the second terms on the right hand sides are absent. In this case $G\backslash f/e=G\backslash e/f$ and we get the result.

Otherwise the edges $h,f,h$ in the graphs $G\backslash f/e$, $G\backslash h/e$, $G\backslash e/f$ form a cycle $C$ attached to $G\backslash ef/h$.
To orient $C$ we reverse the orientation of the edge $h$ in $G\backslash f/e$. This gives a minus sign in $\Psi^{I,J'h}_K(G\backslash f/e)$.
Using (\ref{c0}) on the second superscript gives the result.
\end{proof}

A third family of equations are the Dodgson identities.
\begin{lem}
For any $n\times n$ matrix $M$ and any $i\neq j\in\{1,\ldots,n\}$,
\begin{equation}\label{eqprop1}
\det M^{i,i}\det M^{j,j}-\det M^{i,j}\det M^{j,i}=\det M\det M^{ij,ij},
\end{equation}
where $M^{I,J}$ is $M$ with rows $I$ and columns $J$ deleted.
\end{lem}
\begin{proof}
This is the special case (20) of Lemma 28 in \cite{Bperiods}.
%Assume $\det M\neq0$. Let $m_i$ be the $i$th column of $M$ and $\tilde{m}_i$ be the $i$th column of the adjoint $\widetilde{M}$ of $M$.
%Substitute $\tilde{m}_i$ and $\tilde{m}_j$ in columns  $i$ and $j$ of the $n\times n$ unit matrix $I_n$.
%Call the resulting matrix $M_{i,j}$. From $M\widetilde{M}=\det MI_n$ we obtain
%$$
%MM_{i,j}=(m_1,\ldots,m_{i-1},\det M e_i,m_{i+1},\ldots,m_{j-1},\det M e_j,m_{j+1},\ldots,m_n).
%$$
%Taking determinants on both sides gives the result.
%
%If $\det M=0$ we approximate $M$ by invertible matrices.
\end{proof}

Here, we only need a special Dodgson identity. More general results are in Lemma 30 of \cite{Bperiods}.
\begin{lem}[Dodgson identity]
Let $G$ be a graph, $e,f\in\sE(G)$, $I,J$ words in $\sE(G)$ and $K\subseteq\sE(G)$. Then
\begin{equation}\label{did}
\Psi^{Ie,Je}_{Kf}\Psi^{If,Jf}_{Ke}-\Psi^{I,J}_{Kef}\Psi^{Ief,Jef}_K=\Psi^{Ie,Jf}_K\Psi^{If,Je}_K,
\end{equation}
where every Dodgson is evaluated at the graph $G$.
\end{lem}
\begin{proof}
From (\ref{Ax}) and (\ref{psidef}) we see that all products of Dodgsons in (\ref{did}) have the same sign. It hence suffices to prove (\ref{did})
on the level of determinants. To do this we assume $e,f\notin IJ$ (without restriction), pick $v\in\sV(G)$ and use (\ref{eqprop1}) with $i=e$, $j=f$, $M=L^{Iv,Jv}_{Kef}$.
\end{proof}
Note that the left hand side of (\ref{did}) can be rephrased in terms of minors.

A standard situation is that a graph has several 3-valent vertices. We need to re-prove a result in \cite{Bperiods} with our sign convention.
\begin{lem}[Example 32 in \cite{Bperiods} and Lemma 22 in \cite{K3}]\label{deff}
Let $G$ be a graph and $u\in\sV(G)$ be an oriented vertex of degree 3 with no self-loops. Then, the graph polynomial $\Psi(G)$ has the structure
\begin{equation}\label{3a}
\Psi=f_0(\alpha_1\alpha_2+\alpha_1\alpha_3+\alpha_2\alpha_3)-(f_2+f_3)\alpha_1-(f_1+f_3)\alpha_2-(f_1+f_2)\alpha_3+f_{123},
\end{equation}
where the polynomials
\begin{equation}\label{fdef}
f_0=\Psi^{12,12}_3,\;f_1=\Psi^{2,3}_1,\;f_2=\Psi^{1,3}_2,\;f_3=\Psi^{1,2}_3,\;f_{123}=\Psi_{123}
\end{equation}
fulfill the equation
\begin{equation}\label{3b}
f_0f_{123}=f_1f_2+f_1f_3+f_2f_3.
\end{equation}
\end{lem}
\begin{proof}
We use (\ref{v0a}) for edge $e=1$ yielding $\Psi^{1,2}=-\Psi^{2,2}_1-\Psi^{3,2}_1$. Setting $\alpha_3=0$ we get $\Psi^{2,2}_{13}=-f_1-f_3$.
By contraction-deletion (\ref{cd}), $\Psi^{2,2}_{13}$ is the coefficient of $\alpha_2$ in $\Psi$. Likewise we get the coefficients of $\alpha_1$ and $\alpha_3$.
The constant term in $\alpha_1,\alpha_2,\alpha_3$ is $f_{123}$ whereas the coefficient of $\alpha_1\alpha_2$ is $f_0$.
Because $G\backslash 12/3=G\backslash 13/2=G\backslash 23/1$ the other quadratic terms have the same coefficient.
There is no cubic term because 123 cuts $G$, see Lemma \ref{lemcuts}.

Finally, we use the Dodgson identity (\ref{did}) with $e=1$, $f=2$, $K=\{3\}$ and get
$$
\Psi^{1,1}_{23}\Psi^{2,2}_{13}-\Psi_{123}\Psi^{12,12}_3=(\Psi^{1,2}_3)^2.
$$
With (\ref{3a}) and (\ref{fdef}) this gives $(-f_2-f_3)(-f_1-f_3)-f_{123}f_0=(f_3)^2$ which is (\ref{3b}).
\end{proof}

The degrees of Dodgson polynomials,
\begin{equation}\label{doddeg}
\deg\Psi^{I,J}_K(G)=h_1(G)-|I|,\quad\hbox{if }\Psi^{I,J}_K(G)\neq0,
\end{equation}
follow from (\ref{1}) for $I=\emptyset$. For general $I$, Equation (\ref{doddeg}) is obtained by iteratively using (\ref{cd}) and (\ref{v0a}).

Note that there exist many more identities for Dodgsons in \cite{Bperiods,K3,BSYc2}. There also exists a powerful combinatorial approach to Dodgson polynomials.
This approach relates monomials to spanning forests and is pursued in \cite{BYc2,BSYc2,Yc2,Yc2pre,HSSYc2,Hourglass}.

\section{The $c_2$-invariant}\label{sectc2}
In the next two sections we review the definition and basic properties of the $c_2$-invariant.
Because the material has been covered in several articles (see e.g.\ \cite{K3,HSSYc2} and the references therein) we keep these sections short.
In particular, we do not give proofs but refer to the literature.

\begin{defn}\label{defprim}
A graph $G$ with $h_1(G)$ independent cycles is {\em primitive} if $|\sE(G)|=2h_1(G)$ and every non-empty proper subgraph $g$ of $G$ has $|\sE(g)|>2h_1(g)$.
\end{defn}

A primitive $\phi^4$ graph is short of being 4-regular (i.e.\ every vertex has degree 4) by four half-edges. If one adds an extra vertex $\infty$ to $G$ and connects these four half-edges
to $\infty$ then one obtains a 4-regular graph: the completion $\overline{G}$ of $G$. Conversely, $G$ is a decompletion of $\overline{G}$. While completion is unique, decompletion is not
(in general).

Any primitive graph $G$ has a period (\ref{2}) which is a contribution to the beta-function of four-dimensional $\phi^4$ theory, see e.g.\ \cite{IZ}.
Every decompletion of a completed graph has the same period (Theorem 2.7 in \cite{Scensus}).

The smallest primitive graph is a double edge which has period 1. Periods of larger graphs are non-trivial (and sometimes very hard) to calculate \cite{Snumfunct}.
Lists of known periods are in the files {\tt Periods.m} (for $\phi^4$ and $h_1(G)\leq11$) and {\tt PeriodsNonPhi4} (non-$\phi^4$ with $h_1(G)\leq8$) in \cite{Shlog}.
Although the period exists for all primitive graphs, non-$\phi^4$ periods have no physical meaning. The $\phi^4$ periods have a conjectural coaction structure \cite{PScoaction}
which may indicate a deep connection between quantum field theory and motivic Galois theory \cite{Bcoact1,Bcoact2}.

Because of this motivic connection there exists interest in $\phi^4$ periods from the mathematical as well as from the physical side.
However, progress in calculating $\phi^4$ periods is modest. It is a lucky coincidence that there exists the combinatorial $c_2$-invariant,
which captures some number theoretic aspects of the period. This $c_2$-invariant---while being easier to determine---is able to detect arithmetic geometries which will persist
in the period, driving some of its motivic structure \cite{BD}.

\begin{defn}\label{c2}
Let $G$ be a graph with at least three vertices and let $\FF_q$ be the finite field with $q=p^n$ elements ($p$ prime).
In this case the point-count (\ref{2a}) of the graph polynomial is divisible by $q^2$ (Theorem 2.9 in \cite{SFq} and Proposition 2 in \cite{K3}). The $c_2$ of $G$ at $q$ is
\begin{equation}
c_2^{(q)}(G)\equiv\frac{[\Psi_G]_q}{q^2} \mod q.
\end{equation}
\end{defn}
Note that for a fixed graph $G$ the $c_2$ associates to every prime power $q$ an element in $\ZZ/q\ZZ$. In many cases the $c_2$ is easier to calculate for primes $q=p$ (see Conjecture
\ref{conqp}). In practice, one often has to be even more modest and content oneself with the knowledge of the $c_2$ for a finite prefix of primes (here typically all primes up to 97).

Note that many graphs may have the same period. The $c_2$ should be a period invariant (Conjecture \ref{con2}).
In particular, all decompletions of a completed graph conjecturally have the same $c_2$ (the completion conjecture).

One can use the completion conjecture to limit the number of graphs which may lead to new $c_2$s. Firstly, in $\phi^4$ we only need to consider completed graphs.
Completed 4-regular graphs have primitive decompletions if and only if they are internally 6-edge-connected, i.e.\ the only 4-edge-splits come from cutting off a vertex
(Proposition 2.6 in \cite{Scensus}). Secondly, with Proposition 31 in \cite{K3}, the completion conjecture implies that the $c_2$ vanishes if a completed graph
has a 3-vertex-split. Thirdly, if a completed graph has two triangles $abc$ and $abd$ sharing an edge $ab$ then the $c_2$ is invariant under the {\em double triangle} reduction of
vertex $a$ by the overcrossing $cd$, $be$, where $e$ is the fourth vertex connected to $a$ ($e$ must not be connected to $b$, Theorem 3.5 in \cite{HSSYc2} using
Theorem 35 in \cite{BYc2}). This leads to the following definition.

\begin{defn}[Definition 2.5 in \cite{Scensus}]\label{primeancestor}
A graph with $\geq5$ vertices is a prime ancestor if
\begin{enumerate}
\item it is 4-regular, and
\item it is internally 6-edge-connected, and
\item it has vertex connectivity 4, and
\item it has no edge which is shared by exactly two triangles.
\end{enumerate}
\end{defn}

\begin{con}
The $c_2$-invariants of prime ancestor decompletions exhaust all $c_2$-invariants up to a given loop order.
\end{con}
We will assume the above conjecture and investigate only (single decompletions of) prime ancestors. For non-$\phi^4$ graphs an analogous definition exists which,
however, is less powerful in the sense that vast numbers of prime ancestors give the same $c_2$. In $\phi^4$ the number of prime ancestors with equal $c_2$ at fixed
loop order is usually quite modest (see Conjecture \ref{K5} and Problem \ref{probgraphc2}).

\section{Denominator reduction}\label{sectdr}
Apart from the reduction to prime ancestors from the last section, the main tool for calculating $c_2$s is denominator reduction.
Consider a primitive graph $G$ with at least three edges. Assume we want to calculate the period (\ref{2}) by integrating out one variable after the other
(note that there exist much more powerful tools \cite{Snumfunct}). With (\ref{cd}) and (\ref{did}) we get
\begin{eqnarray*}
P(G)&=&\int_{\alpha>0}\frac{\dd \alpha_1\ldots}{\Psi^2}=\int_{\alpha>0}\frac{\dd \alpha_2\ldots}{\Psi^{1,1}\Psi_1}=
 \int_{\alpha>0}\frac{\log\frac{\Psi^{1,1}_2\Psi^{2,2}_1}{\Psi^{12,12}\Psi_{12}}\dd\alpha_3\ldots}{(\Psi^{1,2})^2}\\
&=&\int_{\alpha>0}\frac{\log\frac{\Psi^{1,1}_{23}\Psi^{2,2}_{13}}{\Psi^{12,12}_3\Psi_{123}}\dd\alpha_4\ldots}{\Psi^{13,23}\Psi^{1,2}_3}+\ldots.
\end{eqnarray*}
In the last equation we used integration by parts and skipped some terms which are of similar structure. Note that we have done three integrations but only obtained logarithms
(no di- or tri-logarithms). This can be considered as a double weight drop and it corresponds to the divisibility of $[\Psi_G]_q$ by $q^2$ in Definition \ref{c2}.
The log in the numerator of the last integrand consists of graph polynomials of minors [by (\ref{cd})]. The most complicated part of the geometry in the integrand comes
from the denominator. This denominator is seen in the $c_2$-invariant.
\begin{lem}[Corollary 28 and Theorem 29 in \cite{K3}]
For any graph $G$ with a 3-valent vertex we have
\begin{equation}\label{c2dod}
c_2^{(q)}(G)\equiv-[\Psi^{13,23}(G)\Psi^{1,2}_3(G)]_q\mod q.
\end{equation}
\end{lem}
Note that any primitive graph with loop order $>1$ has at least four 3-valent vertices.

By inclusion-exclusion the right hand side of (\ref{c2dod}) is $-[\Psi^{13,23}]_q-[\Psi^{1,2}_3]_q+[\Psi^{13,23},\Psi^{1,2}_3]_q$ where the last term is
the point-count of the intersection. Because (by linearity in each variable, see Corollary 2.6 in \cite{SFq}) $[\Psi^{13,23}]_q\equiv[\Psi^{1,2}_3]_q\equiv0\mod q$
the $c_2$ is the point-count of the intersection of two Dodgson polynomials. In general we define:

\begin{defn}\label{defdodpair}
A {\em Dodgson pair} is a pair of homogeneous polynomials in $\ZZ[\alpha_1,\ldots,\alpha_N]$ of degrees $d,N-d>0$ such that
each polynomial is linear in each variable. A Dodgson intersection is the intersection of a Dodgson pair.
\end{defn}

Examples of Dodgson pairs are ($\Psi^{1,1}$, $\Psi_1$) and ($\Psi^{13,23}$, $\Psi^{1,2}_3$). Table \ref{tab3} at the end of the article contains a list of Dodgson pairs with low-dimensional
intersections. The resultant $[\Phi_1,\Phi_2]_\alpha$ of a Dodgson pair $\Phi_1$, $\Phi_2$ with respect
to any variable $\alpha$ parametrizes the Dodgson intersection. It is of degree $N-1$ in $N-1$ variables and has zero locus of dimension $N-3$ as a projective variety.
Dodgson intersections of dimensions 0 and 1 are square roots and elliptic curves, respectively.

One may continue to integrate out variables in the period integral as long as the polynomials factorize to produce a variable which is linear in all factors (see \cite{Bperiods,Phyperint}).
In this case we obtain a sequence of denominators which are quadratic in each variable. If we focus on the denominator the condition for staying clear of
roots is that the denominator factorizes (as it does after three integrations). If each factor is linear in all variables then the denominator is the product of a Dodgson pair.

\begin{defn}[Denominator reduction, Definition 120 and Proposition 126 in \cite{Bperiods}]\label{defdr}
Given a graph $G$ with at least three edges and a sequence of edges $1,2,\ldots,{|\sE(G)|}$ we define
$$
^3\Psi_G(1,2,3)=\pm\Psi^{13,23}(G)\Psi^{1,2}_3(G).
$$
Suppose $^n\Psi_G$ for $n\geq3$ factorizes as
$$
^n\Psi_G(1,\ldots,n)=(A\alpha_{n+1}+B)(C\alpha_{n+1}+D)
$$
then we define
$$
^{n+1}\Psi_G(1,\ldots,n+1)=\pm(AD-BC).
$$
Otherwise denominator reduction terminates at step $n$. If it exists we call $^n\Psi_G$ an $n$-invariant of $G$.
If $^n\Psi_G=0$ for some sequence of edges and some $n$ (and hence for all subsequent $n$) then we say that $G$ has weight drop.
\end{defn}
Note that the $n$-invariants are only defined up to sign.

\begin{thm}[Corollary 125 in \cite{Bperiods}]\label{standarddenomred}
For any graph $G$ with at least five edges all 5-invariants exist. Explicitly, one has
\begin{equation}\label{5invariant}
^5\Psi_G(1,2,3,4,5)=\pm\det\left(\begin{array}{cc}
 \Psi^{125,345}&\Psi^{135,245}\\
 \Psi^{12,34}_5&\Psi^{13,24}_5\end{array}\right).
\end{equation}
For $n\geq5$ all $n$-invariants (if existent) do not depend on the order of edges.
\end{thm}
In Corollary 125 in \cite{Bperiods} it is proved that after five integrations the integrand of the period has the 5-invariant as unique denominator.
It follows that the 5-invariant does not depend on the order of edges (one can also use (\ref{c0}) on the 3-invariant).
For $n>5$ invariance follows by Fubini's theorem and the connection to denominators of the integrand (or by direct computation, see the remark after Definition 14 in \cite{K3}).
Note that for some sequences of edges denominator reduction may terminate sooner than for others. It is not clear in general what an ideal sequence of edges is for getting high invariants.

Denominator reduction is compatible with the $c_2$-invariant:

\begin{thm}[Theorem 29 in \cite{K3}]\label{denredthm}
Let $G$ be a connected graph with at least three edges and $2h_1(G)\leq|\sE(G)|$ then,
\begin{equation}\label{c2frominvariant}
c_2^{(q)}(G)\equiv (-1)^n [^n\Psi_G]_q\mod q
\end{equation}
whenever $^n\Psi_G$ exists for $n<|\sE(G)|$. If $2h_1(G)<|\sE(G)|\geq4$, then $c_2^{(q)}(G)\equiv0\mod q$.
\end{thm}
The above theorem was proved for a minimum of five edges in \cite{K3} (where $^n\Psi_G$ was only defined for $n\geq5$). It trivially extends to the case
of three or four edges.

Note that the point-count $[^n\Psi_G]_q$ is well-defined for all prime powers although the $n$-invariant is only defined up to sign.
Because often denominator reduction is possible for many steps, the above theorem is a powerful tool to determine the $c_2$ for small $q$.
By experiment, however, we know that denominator reduction for prime ancestors almost never provides all possible reductions.

\section{The Legendre symbol for $\FF_q$}\label{sectls}
The Legendre symbol $(a/p)$ is $\pm1$ depending on whether or not $a\neq0$ is a square in $\FF_p$ (and $(0/p)=0$).
We need an analogous definition for $\FF_q$. In this section we assume that $q=p^n$ for an odd prime $p$.
We project $\ZZ$ onto $\FF_p\subseteq\FF_q$ by $a\mapsto a\cdot1$, $1\in\FF_q$. Because $p$ is odd, $-1,0,1$ are mutually distinct in $\FF_q$.
We identify these three integers with their images in $\FF_q$.
\begin{defn}\label{deflege}
For any $a\in\FF_q$ the Legendre symbol $(a/q)\in\{-1,0,1\}$ is defined by
\begin{equation}\label{legedef}
\left(\frac{a}{q}\right)=|\{x\in\FF_q:x^2=a\}|-1.
\end{equation}
For any polynomial $F\in\ZZ[\alpha_1,\ldots,\alpha_N]$ we define
\begin{equation}\label{Fqdef}
(F)_q=\sum_{\alpha\in\FF_q^N}\left(\frac{F(\alpha)}{q}\right),
\end{equation}
where the sum is in $\ZZ$.
\end{defn}

Some elementary properties of the Legendre symbol are summarized in the following lemma.
\begin{lem}
\begin{enumerate}
\item The Legendre symbol is multiplicative,
\begin{equation}\label{legemult}
\left(\frac{ab}{q}\right)=\left(\frac{a}{q}\right)\left(\frac{b}{q}\right).
\end{equation}
\item If $a\in\ZZ$ (and $q=p^n$) the Legendre symbol $(a/q)$ agrees with the Jacobi symbol,
$$
\left(\frac{a}{q}\right)=\left(\frac{a}{p}\right)^n.
$$
\item
For any $a\in\FF_q$,
\begin{equation}\label{legeexpl}
\left(\frac{a}{q}\right)=a^{\frac{q-1}{2}}.
\end{equation}
\item
For any $A\in\FF_q^\times$ and $B\in\FF_q$,
\begin{equation}\label{legesum}
(A\alpha+B)_q=0.
\end{equation}
\item
For any $F\in\ZZ[\alpha_1,\ldots,\alpha_N]$ the point-count $[F]_q$ can be expressed in terms of the Legendre symbol as
\begin{equation}\label{legepc}
[F]_q=q^N-(F^2)_q.
\end{equation}
\end{enumerate}
\end{lem}
\begin{proof}
For multiplicativity it suffices that any $a\in\FF_q^\times$ is in half the cases a square and in half the cases a non-square, so that the product of two non-squares is a square.

For (2) we observe that the unique quadratic extension of $\FF_p$ is $\FF_{p^2}$ which is in $\FF_q$ if and only if $n$ is even.

For (3) we may assume that $a\neq0$. The multiplicative group $\FF_q^\times$ is cyclic. If $b$ is a generator of $\FF_q^\times$ then $a=b^m$ for some integer $m\geq0$.
If $m$ is even then $a$ is the square of $b^{m/2}$ and $(a/q)=1$, otherwise $(a/q)=-1$. Substituting $a=b^m$ into the right hand side of (\ref{legeexpl}) gives the result
because $m(q-1)/2\equiv0\mod q-1$ if and only if $m$ is even.

For (4) we transform $\alpha$ in the sum by the bijection $\alpha\mapsto(\alpha-B)/A$ yielding $(A\alpha+B)_q=(\alpha)_q$. We have $(0/q)=0$.
For $\alpha\in\FF_q^\times$, half the $\alpha$ have $(\alpha/q)=1$ while the other half has $(\alpha/q)=-1$.

For (5) we observe that $(F^2/q)=1-\chi(F)$ with the characteristic function $\chi(F)$ being 1 if $F=0$ in $\FF_q$ and 0 otherwise. The result follows by summing over
$\alpha\in\FF_q^N$.
\end{proof}

Equation (\ref{legeexpl}) gives a Chevalley-Warning method to calculate $(F)_q$ modulo $p$ (see e.g.\ \cite{Sa}).
\begin{lem}\label{lemsumx}
Let $1\leq k\leq q-2$. Then
\begin{equation}\label{sumx}
\sum_{\alpha\in\FF_q}\alpha^k=0.
\end{equation}
\end{lem}
\begin{proof}
Because $k\geq1$ we can drop 0 from the sum and write $\alpha=b^m$ for a generator $b$ of $\FF_q^\times$. The sum becomes a geometric series where $m$ ranges from 0 to $q-2$.
Because $k\leq q-2$ we have $b^k\neq1$ and summing up the series provides a numerator $b^{k(q-1)}-1$. From $b^{q-1}=1$ we deduce the result.
\end{proof}

\begin{lem}[Chevalley-Warning]\label{lemcw}
Let $F\in\ZZ[\alpha_1,\ldots,\alpha_N]$ be of degree $\leq 2N$. Then
\begin{equation}\label{chewar}
(-1)^N(F)_q\equiv\hbox{coefficient of $(\alpha_1\cdots\alpha_N)^{q-1}$ in }F^\frac{q-1}{2}\mod p.
\end{equation}
In particular, $(F)_q\equiv0\mod p$ if $\deg(F)<2N$.
\end{lem}
\begin{proof}
We use (\ref{legeexpl}) to calculate $(F)_q$. We project the sum (\ref{Fqdef}) onto $\FF_p\subseteq\FF_q$ which amounts to calculating modulo $p$. We use (\ref{legeexpl})
and expand $F^{(q-1)/2}$. With $\sum_{\alpha\in\FF_q}1=q\equiv0\mod p$ and Lemma \ref{lemsumx} we see that we need at least an exponent $q-1$ in each variable to obtain a non-vanishing term.
The unique term of lowest degree with this property is $c(\alpha_1\cdots\alpha_N)^{q-1}$, for some coefficient $c\in\ZZ$.
Because the degree of $F^{(q-1)/2}$ is at most $N(q-1)$ this term is the only contribution to the sum. The sum over this monomial can be restricted to $\FF_q^\times$.
For any $\alpha_i\in\FF_q^\times$ we have $\alpha_i^{q-1}=1$. Hence the sum gives $c(q-1)^N\equiv c(-1)^N\mod p$. If $\deg(F)<2N$ then $c=0$.
\end{proof}
Equation (\ref{chewar}) does not hold modulo $q$. A simple counter-example is $F=(2\alpha)^2$ and $q=9$. We get $-(F)_9\equiv[2\alpha]_9\equiv1\mod 9$ whereas the $\alpha^8$ coefficient of $F^4$
is $2^8\equiv4\mod9$. In this article we solely need the second statement of the lemma which extends to prime powers.
\begin{cor}[of Theorem \ref{thm:main}]\label{corcw}
Let $N\geq1$ and $F\in\ZZ[\alpha_1,\ldots,\alpha_N]$ be of degree $<2N$. Then $(F)_q\equiv0\mod q$ for any odd prime power $q$.
\end{cor}
\begin{proof}
In the notation of Theorem \ref{thm:main} we get $0\equiv|X(k)|=[y^2-F]_q=(F)_q+q^N\equiv(F)_q\mod q$, where we summed (\ref{legedef}) over $\alpha\in\FF_q^N$.
\end{proof}

\section{Quadratic denominator reduction}\label{sectionqdr}
The next integration after denominator reduction terminates is schematically (ignoring numerators in the integrand)
$$
\int_0^\infty\frac{\dd\alpha}{A\alpha^2+B\alpha+C}=\frac{\log(X)}{\sqrt{B^2-4AC}},
$$
for some algebraic $X$. The discriminant $B^2-4AC$ will be of degree $\leq4$ in each variable.
For the next integration we may extract square factors in the root and only enter the elliptic setup if thereafter the argument of the root is of degree $\geq3$ in each variable.
Otherwise we schematically have the structure
$$
\int_0^\infty\frac{\dd\alpha}{\sqrt{D\alpha^2+E\alpha+F}(H\alpha+J)}=\frac{\log(Y)}{\sqrt{DJ^2-EHJ+FH^2}}.
$$
The root on the right hand side may be seen as residue of the integrand at $\alpha=-J/H$.
The general idea of quadratic denominator reduction is to use these structures to continue eliminating variables as long as the geometry of the denominator remains rational.

\begin{defn}[Quadratic denominator reduction]\label{defqdr}
Given a graph $G$ with at least three edges and a sequence of edges $1,2,\ldots,{|\sE(G)|}$ we define
\begin{equation}\label{n3}
^3\Psi^2_G(1,2,3)=[\pm^3\Psi_G(1,2,3)]^2=[\Psi^{13,23}(G)\Psi^{1,2}_3(G)]^2.
\end{equation}
Suppose $^n\Psi^2_G$ for $n\geq3$ is of the form
\begin{equation}\label{case1}
^n\Psi^2_G(1,\ldots,n)=(A\alpha_{n+1}^2+B\alpha_{n+1}+C)^2
\end{equation}
then we define
$$
^{n+1}\Psi^2_G(1,\ldots,n+1)=B^2-4AC.
$$
Suppose $^n\Psi^2_G$ is of the form
\begin{equation}\label{case2}
^n\Psi^2_G(1,\ldots,n)=(D\alpha_{n+1}^2+E\alpha_{n+1}+F)(H\alpha_{n+1}+J)^2
\end{equation}
then we define
$$
^{n+1}\Psi^2_G(1,\ldots,n+1)=DJ^2-EHJ+FH^2.
$$
Otherwise quadratic denominator reduction terminates at step $n$. If it exists we call $^n\Psi^2_G$ a quadratic $n$-invariant of $G$.
If $^n\Psi^2_G=0$ for some sequence of edges and some $n$ then we say that $G$ has weight drop.
\end{defn}
Note that quadratic $n$-invariants have no sign ambiguity. For primitive graphs they have degree $2(|\sE(G)|-n)$ [see (\ref{doddeg})].
Compatibility of the cases (\ref{case1}) and (\ref{case2}) with standard denominator reduction is proved in the next lemma (also see Proposition 126 in \cite{Bperiods}).

\begin{lem}\label{lemquadratic}
In the case of both structures (\ref{case1}) and (\ref{case2}) standard denominator reduction can be used for $^n\Psi_G:=\pm(^n\Psi_G^2)^{1/2}$. Any of the two quadratic reductions leads to
$$
^{n+1}\Psi^2_G(1,\ldots,n+1)=[\pm^{n+1}\Psi_G(1,\ldots,n+1)]^2.
$$
\end{lem}
\begin{proof}
If a perfect square (\ref{case1}) of degree 4 in $\alpha_{n+1}$ factorizes according to (\ref{case2}) then it is of the form
$$
(a\alpha_{n+1}+b)^2(c\alpha_{n+1}+d)^2.
$$
In case (\ref{case1}) we have $A=ac$, $B=ad+bc$, and $D=bd$ leading to
$$
^{n+1}\Psi^2_G(1,\ldots,n+1)=(ad+bc)^2-4acbd=(ad-bc)^2.
$$
In case (\ref{case2}) we have the ambiguity of moving a common square constant in the first factor to the second factor and vice versa.
The reduction being of degree $(1,2)$ in $D,E,F$ and $H,J$ is compatible with this ambiguity. Moreover, standard denominator reduction is symmetric under swapping the
two linear factors. We can hence restrict ourselves to the case $D=a^2$, $E=2ab$, $F=b^2$, $H=c$, $J=d$ in which we obtain
$$
^{n+1}\Psi^2_G(1,\ldots,n+1)=a^2d^2-2abcd+b^2c^2=(ad-bc)^2.
$$
\end{proof}

Although one might expect that case (\ref{case2}) is rare in practice, by experiment we find that rather the contrary is true.
With each loop order a small but increasing number of prime ancestors even seems to reduce to a constant. In \cite{Hourglass}, K. Yeats and the author prove full reductions for families of
$\phi^4$ ancestors.

The connection to the $c_2$-invariant is similar to the standard case.
\begin{thm}\label{thmquadratic}
Let $q$ be an odd prime power. Let $G$ be a connected graph with at least three edges and $2h_1(G)\leq|\sE(G)|$ then
\begin{equation}\label{c2fromquadratic}
c_2^{(q)}(G)\equiv (-1)^{n-1}(^n\Psi^2_G)_q\mod q
\end{equation}
whenever $^n\Psi^2_G$ exists. If $2h_1(G)<|\sE(G)|\geq4$, then $c_2^{(q)}(G)\equiv0\mod q$.
\end{thm}
If the standard $n$-invariant exists we get the statement of the theorem by Lemma \ref{lemquadratic} and (\ref{legepc}).
\begin{remark}
By Chevalley-Warning, Theorem \ref{thmquadratic} extends to the prime $2$:
In characteristic 2, quadratic denominator reduction picks (squares of) linear coefficients in (squares of) quadratic polynomials, see (\ref{case1}) and Lemma \ref{lemquadratic}.
If we define $(a/2)=1-\delta_{a,0}$ for $a\in\FF_2$, we get $(F^2)_2=(F)_2=2^N-[F]_2$ for $F\in\ZZ[\alpha_1,\ldots,\alpha_N]$.
The connection (\ref{legepc}) between the point-count and the Legendre sum stays intact and quadratic denominator reduction picks the coefficient of $\alpha_1\cdots\alpha_N$
in $F$.

A combinatorial point-count for the prime 2 (and beyond) is performed for some families of prime ancestors in \cite{Yc2,Yc2pre}.

It is unclear if Theorem \ref{thmquadratic} extends to all prime powers if one defines $(a/2^n)=1-\delta_{a,0}$ for $a\in\FF_{2^n}$.
\end{remark}

\section{Proof of Theorem \ref{thmquadratic}}\label{sectpf}
Throughout this section $q=p^n$ is an odd prime power. We need a sequence of lemmas.
\begin{lem}
For any $A,B,C\in\FF_q$,
\begin{equation}\label{pc2ls}
[A\alpha^2+B\alpha+C]_q\equiv\left(\frac{B^2-4AC}{q}\right)+\left(\frac{A^2}{q}\right)\mod q.
\end{equation}
\end{lem}
\begin{proof}
For the point-count we get $q$ if $A=B=C=0$, 0 if $A=B=0$, $C\neq0$, and 1 ($\alpha=-C/B$) if $A=0$, $B\neq0$. In any case (\ref{pc2ls}) holds.

If $A\neq0$ then the left hand side equals $|\{x\in\FF_q:x^2=B^2-4AC\}|$. Because $(A^2/q)=1$ we get the result from (\ref{legedef}).
\end{proof}

The next lemma follows modulo $p$ from Lemma \ref{lemcw} if we assume that the constants $A,B,C$ are integers. Here, we need the general case.
\begin{lem}\label{lem2}
For any $A,B,C\in\FF_q$,
\begin{equation}\label{prop2eq}
(A\alpha^2+B\alpha+C)_q\equiv-\left(\frac{A}{q}\right)\mod q.
\end{equation}
\end{lem}
\begin{proof}
First assume $A=1,B=0$. If $C=0$ then $(\alpha^2)_q=q-1$ and the result follows.

If $(C/q)=-1$, then
$$
(\alpha^2+C)_q=\left(\frac{C}{q}\right)+\sum_{\alpha\in\FF_q^\times}\left(\frac{\alpha^2+C}{q}\right)=-1+\sum_{\alpha\in\FF_q^\times}\left(\frac{(C/\alpha)^2+C}{q}\right),
$$
where we used the bijection $\alpha\mapsto C/\alpha$ on $\FF_q^\times$. Using multiplicativity (\ref{legemult}) and $(\alpha^{-2}/q)=1$ for $\alpha\in\FF_q^\times$ we obtain for the
sum on the right hand side
$$
\sum_{\alpha\in\FF_q^\times}\left(\frac{C/\alpha^2}{q}\right)\left(\frac{C+\alpha^2}{q}\right)=-\sum_{\alpha\in\FF_q^\times}\left(\frac{C+\alpha^2}{q}\right)=-1-(\alpha^2+C)_q.
$$
Solving for $(\alpha^2+C)_q$ gives $-1$ and (\ref{prop2eq}) holds.

If $(C/q)=1$ then there exists an $x\in\FF_q^\times$ such that $x^2=C$. With the bijection $\alpha\mapsto\alpha x$ we find $(\alpha^2+C)_q=(\alpha^2+1)_q$.

Now, we sum over all $C$, interchange the sums over $C$ and $\alpha$, and use (\ref{legesum}) to see that
$$
0=\sum_{C\in\FF_q}(\alpha^2+C)_q=q-1+\frac{q-1}{2}(-1)+\frac{q-1}{2}(\alpha^2+1)_q,
$$
where we used the previous results. So, $(\alpha^2+1)_q=-1$ as desired.

To handle the case $A=1,B\neq0$ we substitute $\alpha\mapsto\alpha-B/2$ which eliminates the linear term.

For general nonzero $A$ we use multiplicativity (\ref{legemult}) yielding
$$
(A\alpha^2+B\alpha+C)_q=\left(\frac{A}{q}\right)(\alpha^2+B\alpha/A+C/A)_q\equiv-\left(\frac{A}{q}\right)\mod q.
$$
For $A=0$, $B\neq0$ the result follows from (\ref{legesum}), while $A=B=0$ gives $\sum_\alpha(C/q)=q(C/q)\equiv0\mod q$.
\end{proof}

\begin{lem}\label{lem3}
For any $H\in\FF_q^\times$ and $D,E,F,J\in\FF_q$,
\begin{equation}
((D\alpha^2+E\alpha+F)(H\alpha+J)^2)_q\equiv-\left(\frac{D}{q}\right)-\left(\frac{DJ^2-EHJ+FH^2}{q}\right)\mod q.
\end{equation}
\end{lem}
\begin{proof}
The term $\alpha=-J/H$ in the sum over $\alpha$ vanishes. If we omit this term and use multiplicativity (\ref{legemult}) we get for the left hand side
$$
\sum_{\alpha\neq -J/H}\left(\frac{D\alpha^2+E\alpha+F}{q}\right)=(D\alpha^2+E\alpha+F)_q-\left(\frac{D\frac{J^2}{H^2}-E\frac{J}{H}+F}{q}\right).
$$
Using multiplicativity again we get the result from (\ref{prop2eq}).
\end{proof}

With these preparations we are ready to prove Theorem \ref{thmquadratic}.
\begin{proof}[Proof of Theorem \ref{thmquadratic}]
At the end of Section \ref{sectionqdr} we already saw that the theorem holds if standard denominator reduction exists (i.e.\ at least for $n\leq5$).
Because $2h_1(G)\leq|\sE(G)|$ we deduce from (\ref{doddeg}) and (\ref{n3}) that $\deg(^3\Psi^2_G)\leq2(|\sE(G)|-3)$.
Every elimination step with non-zero result reduces the degree of the quadratic $n$-invariant by 2. By induction we get
\begin{equation}\label{ninvdeg}
\deg(^n\Psi^2_G)\leq2(|\sE(G)|-n).
\end{equation}
Assume we have $n\geq3$ in the situation of case (\ref{case1}). By induction and (\ref{legepc}), (\ref{pc2ls}) we get modulo $q$,
\begin{eqnarray*}
(-1)^{n-1}c_2^{(q)}(G)&\!\equiv&\!(^n\Psi^2_G(1,\ldots,n))_q\;=\;((A\alpha_{n+1}^2+B\alpha_{n+1}+C)^2)_q\\
&\!\equiv&\!-[A\alpha_{n+1}^2+B\alpha_{n+1}+C]_q\;\equiv\;-(B^2-4AC)_q-(A^2)_q.
\end{eqnarray*}
If $n+1=|\sE(G)|$ then, by (\ref{ninvdeg}), we have $A=0$.
Otherwise $A^2$ is of degree $\leq2(|\sE(G)|-n)-4$ in $|\sE(G)|-n-1$ variables.
By Corollary \ref{corcw} we get $(A^2)_q\equiv0\mod q$ and the result follows.
(We also have $(A^2)_q\equiv-[A]_q\equiv0\mod q$ by Ax's extension of the Chevalley-Warning theorem \cite{Ax}.)

In the situation of (\ref{case2}) we have by induction
$$
(-1)^{n-1}c_2^{(q)}(G)\equiv((D\alpha_{n+1}^2+E\alpha_{n+1}+F)(H\alpha_{n+1}+J)^2)_q\mod q.
$$
In the case $H\neq0$ we can use Lemma \ref{lem3} and obtain $-(D)_q-(DJ^2-EHJ+FH^2)_q\mod q$.
If $n+1=|\sE(G)|$ then $D=0$. Otherwise $\deg(D)\leq2(|\sE(G)|-n)-4$, so that $(D)_q\equiv0\mod q$ by Corollary \ref{corcw}.

If $H=0$ then, by Lemma \ref{lem2}, we have $(-1)^{n-1}c_2^{(q)}(G)\equiv-(DJ^2)_q\mod q$ which completes the proof.
\end{proof}

\section{Initial reductions}\label{sectini}
While it is not clear what is the best sequence of edges for denominator reduction in general, it always seems advantageous to begin with 3-valent vertices.
Note that every (non-trivial) decompleted $\phi^4$ prime ancestor has four vertices of degree 3, while primitive non-$\phi^4$ graphs have at least five such vertices.
We immediately see from (\ref{5invariant}) and Lemma \ref{lemcuts} that the 5- and the 6-invariants factorize if one reduces the edges 2,b,3,c,1,a of two generic 3-valent vertices
(see and Figure \ref{fig1} for the labeling of edges). From (\ref{cd}) and Definition \ref{defdr} we obtain
\begin{equation}\label{56inv}
^5\Psi_G=\pm\Psi^{2b1,3c1}\Psi^{23,bc}_1,\quad ^6\Psi_G=\pm\Psi^{2b1a,3c1a}\Psi^{23,bc}_{1a}.
\end{equation}
Moreover note that from (\ref{cd}) and (\ref{23u}) we have
\begin{equation}\label{56invG}
\Psi^{2b1a,3c1a}(G)=\Psi^{2b,3c}(G\backslash1a)=\pm\Psi(G\backslash12ab/3c).
\end{equation}
The graph $G\backslash12ab/3c$ is $G$ with the six edges $1,2,3,a,b,c$ removed (see Lemma 4.5 in \cite{HSSYc2}).
Because the 6-invariant factorizes we get a minimum of seven standard denominator reductions.
We get one extra quadratic reduction by (\ref{case1}). Lemma \ref{leminiqred} shows that we always get nine reductions by eliminating a third 3-valent vertex.

Triangles in the completed graph further simplify denominator reduction. Because we can restrict ourselves to prime ancestors (see Definition \ref{primeancestor}) we ignore the
case of double triangles (two triangles attached at a common edge). This leads to a classification of prime ancestors by three classes of increasing difficulty: graphs with
(1) no triangles, (2) isolated triangles, (3) vertex-attached triangles (hourglasses). Decompletion of these cases lead to the substructures in Figure \ref{fig1}.

\begin{figure}
\begin{center}
\fcolorbox{white}{white}{
  \begin{picture}(274,222) (3,-3)
    \SetWidth{0.8}
    \SetColor{Black}
    \Vertex(30.312,186.421){2}
    \Vertex(103.062,186.421){2}
    \Vertex(175.812,186.421){2}
    \Line(30.312,216.734)(30.312,186.421)
    \Line(30.312,186.421)(6.062,168.234)
    \Line(30.312,186.421)(54.562,168.234)
    \Line(103.062,186.421)(127.312,168.234)
    \Line[arrow,arrowpos=0.5,arrowlength=5,arrowwidth=2,arrowinset=0.2](175.812,186.421)(200.062,168.234)
    \Line(103.062,186.421)(78.812,168.234)
    \Line(103.062,216.734)(103.062,186.421)
    \Line[arrow,arrowpos=0.5,arrowlength=5,arrowwidth=2,arrowinset=0.2,flip](175.812,216.734)(175.812,186.421)
    \Line[arrow,arrowpos=0.5,arrowlength=5,arrowwidth=2,arrowinset=0.2](175.812,186.421)(151.562,168.234)
    \Line(30.312,131.859)(30.312,101.546)
    \Line(30.312,101.546)(54.562,83.359)
    \Line(54.562,150.046)(30.312,131.859)
    \Line(6.062,150.046)(30.312,131.859)
    \Line(30.312,101.546)(6.062,83.359)
    \Line(103.062,119.734)(78.812,101.546)
    \Line(103.062,119.734)(127.312,101.546)
    \Line(103.062,150.046)(103.062,119.734)
    \Line[arrow,arrowpos=0.5,arrowlength=5,arrowwidth=2,arrowinset=0.2,flip](175.812,150.046)(175.812,119.734)
    \Line[arrow,arrowpos=0.5,arrowlength=5,arrowwidth=2,arrowinset=0.2](175.812,119.734)(151.562,101.546)
    \Line[arrow,arrowpos=0.5,arrowlength=5,arrowwidth=2,arrowinset=0.2](175.812,119.734)(200.062,101.546)
    \Line(54.562,65.172)(30.312,46.984)
    \Line(6.062,65.172)(30.312,46.984)
    \Line(30.312,46.984)(30.312,16.672)
    \Line(30.312,16.672)(54.562,-1.516)
    \Line(30.312,16.672)(6.062,-1.516)
    \Line(103.062,16.672)(78.812,-1.516)
    \Line(103.062,16.672)(127.312,-1.516)
    \Line(103.062,46.984)(103.062,16.672)
    \Line(78.812,65.172)(103.062,46.984)
    \Line(127.312,65.172)(103.062,46.984)
    \Line[arrow,arrowpos=0.5,arrowlength=5,arrowwidth=2,arrowinset=0.2](175.812,34.859)(151.562,16.672)
    \Line[arrow,arrowpos=0.5,arrowlength=5,arrowwidth=2,arrowinset=0.2](175.812,34.859)(200.062,16.672)
    \Line[arrow,arrowpos=0.5,arrowlength=5,arrowwidth=2,arrowinset=0.2,flip](175.812,65.172)(175.812,34.859)
    \Vertex(175.812,34.859){2}
    \Vertex(175.812,119.734){2}
    \Vertex(103.062,119.734){2}
    \Vertex(103.062,46.984){2}
    \Vertex(30.312,131.859){2}
    \Vertex(30.312,101.546){2}
    \Vertex(103.062,16.672){2}
    \Vertex(30.312,46.984){2}
    \Vertex(30.312,16.672){2}
    \Text(42,180)[lb]{$1$}
    \Text(33,199)[lb]{$2$}
    \Text(12,180)[lb]{$3$}
    \Text(33,116)[lb]{$1$}
    \Text(42,148)[lb]{$2$}
    \Text(12,148)[lb]{$3$}
    \Text(12,95)[lb]{$4$}
    \Text(42,95)[lb]{$5$}
    \Text(33,30)[lb]{$1$}
    \Text(42,62)[lb]{$2$}
    \Text(12,62)[lb]{$3$}
    \Text(12,10)[lb]{$4$}
    \Text(42,10)[lb]{$5$}
    \Text(116,180)[lb]{$a$}
    \Text(107,199)[lb]{$b$}
    \Text(85,180)[lb]{$c$}
    \Text(116,113)[lb]{$a$}
    \Text(107,134)[lb]{$b$}
    \Text(85,113)[lb]{$c$}
    \Text(107,30)[lb]{$a$}
    \Text(116,62)[lb]{$b$}
    \Text(85,62)[lb]{$c$}
    \Text(85,10)[lb]{$d$}
    \Text(116,10)[lb]{$e$}
    \Text(189,180)[lb]{$A$}
    \Text(179,199)[lb]{$B$}
    \Text(157,180)[lb]{$C$}
    \Text(189,113)[lb]{$A$}
    \Text(179,134)[lb]{$B$}
    \Text(157,113)[lb]{$C$}
    \Text(189,28)[lb]{$A$}
    \Text(179,50)[lb]{$B$}
    \Text(157,28)[lb]{$C$}
    \Text(240,182)[lb]{(1)}
    \Text(240,115)[lb]{(2)}
    \Text(240,30)[lb]{(3)}
  \end{picture}
}
\end{center}
\caption{Some three-valent vertices from the decompletion of (1) a generic 4-valent vertex, (2) a vertex of a triangle, (3) the middle vertex of an hourglass,
where the vertex $ABC$ only exists if the middle vertex has degree $\geq5$. Because ancestors have no double triangles the depicted edges are distinct.}
\label{fig1}
\end{figure}
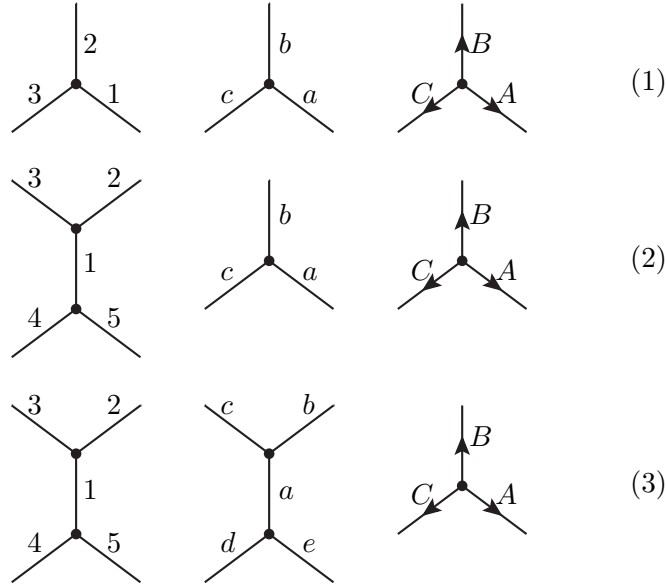

\begin{lem}[Lemma 4.5 in \cite{HSSYc2}]\label{leminired}
Consider a graph $G$ which has a substructure as in Figure \ref{fig1} without the oriented vertex $ABC$. Let $G_0$ be the graph without the plotted edges and vertices.
In the cases $(1)$, $(2)$, $(3)$ we define $H\in\ZZ[\sM(G)]$ as
\begin{eqnarray*}
(1)&&H=G/1a,\\
(2)&&H=G\backslash4/15a-G\backslash5/14a,\\
(3)&&H=G\backslash4d/15ae-G\backslash5d/14ae-G\backslash4e/15ad+G\backslash5e/14ad.
\end{eqnarray*}
Then, the $n$-invariant for $n=6$, $8$, $10$ (respectively) is
\begin{equation}\label{inired}
^n\Psi_G=\pm\Psi(G_0)\Psi^{23,bc}(H).
\end{equation}
\end{lem}
\begin{proof}
We reproduce the proof in \cite{HSSYc2} to accommodate our sign convention.

Case (1) is (\ref{56inv}) and (\ref{56invG}). For case (2) we observe the following structure:
\begin{eqnarray*}
\Psi^{23,bc}(G/1a)&=&\Psi^{23,bc}(G\backslash4/15a)\alpha_4+\Psi^{23,bc}(G\backslash5/14a)\alpha_5+X,\\
\Psi(G\backslash12ab/3c)&=&\Psi(G_0)(\alpha_4+\alpha_5)+Y
\end{eqnarray*}
for some $X,Y$ which are constant in $\alpha_4,$ $\alpha_5$. For the first equation we used (\ref{cd}) and Lemma \ref{lemcuts} for the cut 2, 3, 4, 5.
For the second equation we used the cut 1, 4, 5 and $G\backslash124ab/35c=G\backslash125ab/34c=G_0$. The result follows by standard denominator reduction in $\alpha_4$ and $\alpha_5$.

Case (3) follows from case (2) using edges $d$, $e$ in the same way as (2) followed from (1) using edges 4, 5.
\end{proof}

\begin{lem}\label{leminiqred}
Consider a graph $G$ with a substructure as in Figure \ref{fig1}. With the notation of Lemma \ref{leminired} we define
$f_0,$ $f_A,$ $f_B,$ $f_C,$ $f_{ABC}$ according to (\ref{fdef}) with respect to the edges $A$, $B$, $C$ of the graph $G_0$. Let
\begin{equation}\label{gdef}
g_0=\Psi^{23AB,bcAB}_C,\;g_A=\Psi^{23A,bcA}_{BC},\;g_B=\Psi^{23B,bcB}_{AC},\;g_C=\Psi^{23C,bcC}_{AB},
\end{equation}
all evaluated at $H$, be the coefficients of $\alpha_A\alpha_B$, $\alpha_A$, $\alpha_B$, $\alpha_C$ in $\Psi^{23,bc}(H)$, respectively.
Then, the quadratic $n$-invariant for $n=9$, $11$, $13$ (respectively) is
\begin{equation}\label{iniqred}
^n\Psi_G^2=f_0[f_0\lambda(g_A,g_B,g_C)-4g_0(f_Ag_A+f_Bg_B+f_Cg_C+f_{ABC}g_0)],
\end{equation}
where the K\"all\'en function $\lambda$ is defined as
\begin{equation}\label{Kaellen}
\lambda(x,y,z)=(x-y-z)^2-4yz.
\end{equation}
\end{lem}
Note that case (3) needs more than four 3-valent vertices and hence only exists in non-$\phi^4$ graphs.
\begin{proof}
With (\ref{cd}) we see that $g_0=\Psi^{23,bc}(H\backslash AB/C)$ is also the coefficient of $\alpha_A\alpha_C$ and $\alpha_B\alpha_C$ in $\Psi^{23,bc}(H)$.
By Lemmas \ref{lemcuts}, \ref{deff} the $n$-invariant $\pm^n\Psi_G$ in (\ref{inired}) is
\begin{equation*}
\begin{split}
&\big[f_0(\alpha_A\alpha_B+\alpha_A\alpha_C+\alpha_B\alpha_C)-(f_B+f_C)\alpha_A-(f_A+f_C)\alpha_B-(f_A+f_B)\alpha_C\\
&\quad+\;f_{ABC}\big]\big[g_0(\alpha_A\alpha_B+\alpha_A\alpha_C+\alpha_B\alpha_C)+g_A\alpha_A+g_B\alpha_B+g_C\alpha_C+X\big],
\end{split}
\end{equation*}
with some polynomial $X$ which is constant in $\alpha_A,$ $\alpha_B,$ $\alpha_C$.
Quadratic denominator reduction with (\ref{case1}) of $[\pm^n\Psi_G]^2$ allows us to eliminate $\alpha_A$ and $\alpha_B$ yielding an expression which is quadratic in $\alpha_C$.
This corresponds to the case $H=0$ in (\ref{case2}). The reduction in $\alpha_C$ is the quadratic coefficient
$$
f_0^2\lambda(g_A,g_B,g_C)-4f_0g_0(f_Ag_A+f_Bg_B+f_Cg_C)-4g_0^2(f_Af_B+f_Af_C+f_Bf_C).
$$
With (\ref{3b}) we obtain the result.
\end{proof}

Of the 1731 prime ancestors at loop order 11, only 31 have no triangles. For these only the nine reductions of case (1) in Lemma \ref{leminiqred} are available.
However, in a special setup where one has an arrangement of eight or twelve squares (see Figure \ref{fig2}) one obtains an extra simplification leading to a tenth reduction.
At 11 loops, 8 prime ancestors with no triangle are of this type.

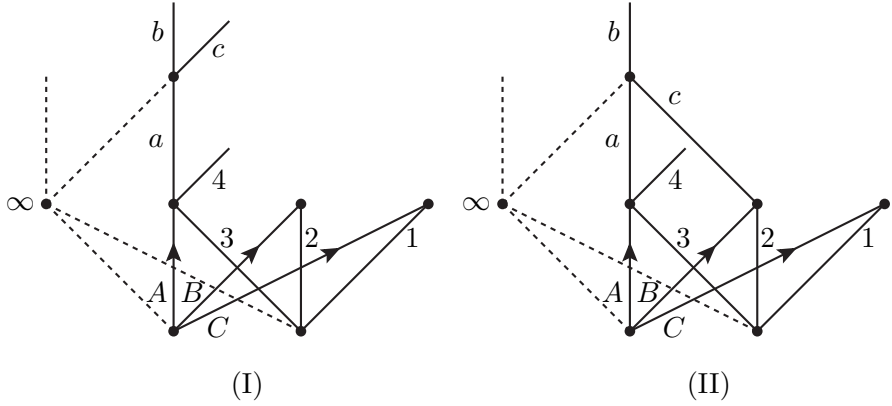
\begin{figure}
\begin{center}
\fcolorbox{white}{white}{
  \begin{picture}(160,155) (70,-15)
    \SetWidth{0.8}
    \SetColor{Black}
    \Vertex(80,59){2}
    \Vertex(128,11){2}
    \Vertex(128,59){2}
    \Vertex(176,59){2}
    \Vertex(224,59){2}
    \Vertex(176,11){2}
    \Vertex(128,107){2}
    \Line[arrow,arrowpos=0.38,arrowlength=5,arrowwidth=2,arrowinset=0.2,flip](176,59)(128,11)
    \Line(128,59)(176,11)
    \Line(176,59)(176,11)
    \Line(224,59)(176,11)
    \Line(128,107)(128,59)
    \Line[dash,dashsize=2](80,59)(128,107)
    \Line[dash,dashsize=2](80,59)(128,11)
    \Line[dash,dashsize=2](80,59)(80,107)
    \Line(128,107)(128,135)
    \Line(128,107)(149,128)
    \Line(128,59)(149,80)
    \Line(176,59)(165,70)
    \Line(176,59)(187,70)
    \Line(224,59)(213,70)
    \Line(224,59)(235,70)
    \Line[arrow,arrowpos=0.62,arrowlength=5,arrowwidth=2,arrowinset=0.2](128,11)(224,59)
    \Line[dash,dashsize=2](80,59)(176,11)
    \Line[arrow,arrowpos=0.38,arrowlength=5,arrowwidth=2,arrowinset=0.2,flip](128,59)(128,11)
    \Text(65,57)[lb]{$\infty$}
    \Text(216,43)[lb]{$1$}
    \Text(178,43)[lb]{$2$}
    \Text(146,43)[lb]{$3$}
    \Text(118,22)[lb]{$A$}
    \Text(131,22)[lb]{$B$}
    \Text(141,9)[lb]{$C$}
    \Text(119,81)[lb]{$a$}
    \Text(120,121)[lb]{$b$}
    \Text(143,115)[lb]{$c$}
    \Text(143,65)[lb]{$4$}
    \Text(150,-15)[lb]{(I)}
  \end{picture}
}\quad
\fcolorbox{white}{white}{
  \begin{picture}(160,155) (70,-15)
    \SetWidth{0.8}
    \SetColor{Black}
    \Vertex(80,59){2}
    \Vertex(128,11){2}
    \Vertex(128,59){2}
    \Vertex(176,59){2}
    \Vertex(224,59){2}
    \Vertex(176,11){2}
    \Vertex(128,107){2}
    \Line[arrow,arrowpos=0.38,arrowlength=5,arrowwidth=2,arrowinset=0.2,flip](176,59)(128,11)
    \Line(128,59)(176,11)
    \Line(176,59)(176,11)
    \Line(224,59)(176,11)
    \Line(128,107)(128,59)
%    \Line(128,107)(139,118)
    \Line[dash,dashsize=2](80,59)(128,107)
    \Line[dash,dashsize=2](80,59)(128,11)
    \Line[dash,dashsize=2](80,59)(80,107)
    \Line(128,107)(128,135)
    \Line(128,59)(149,80)
    \Line(176,59)(197,80)
    \Line(224,59)(213,70)
    \Line(224,59)(235,70)
    \Line(128,107)(176,59)
    \Line[arrow,arrowpos=0.62,arrowlength=5,arrowwidth=2,arrowinset=0.2](128,11)(224,59)
    \Line[dash,dashsize=2](80,59)(176,11)
    \Line[arrow,arrowpos=0.38,arrowlength=5,arrowwidth=2,arrowinset=0.2,flip](128,59)(128,11)
    \Text(65,57)[lb]{$\infty$}
    \Text(216,43)[lb]{$1$}
    \Text(178,43)[lb]{$2$}
    \Text(146,43)[lb]{$3$}
    \Text(118,22)[lb]{$A$}
    \Text(131,22)[lb]{$B$}
    \Text(141,9)[lb]{$C$}
    \Text(119,81)[lb]{$a$}
    \Text(120,121)[lb]{$b$}
    \Text(143,97)[lb]{$c$}
    \Text(143,65)[lb]{$4$}
    \Text(191,65)[lb]{$d$}
    \Text(150,-15)[lb]{(II)}
  \end{picture}
}
\end{center}
\caption{Decompletions of arrangements of eight squares (I) and twelve squares (II).}
\label{fig2}
\end{figure}

\begin{lem}\label{lemsquares}
Assume a graph $G$ has a substructure as in Figure \ref{fig2} (I) or (II).
Let $H=G/1a$ and $G_0$ be defined as in Lemma \ref{leminired} (1) and let $H_0\in\ZZ[\sM(G)]$ be $H$ without the vertex $ABC$.
Then the quadratic 10-invariant with respect to the 10 solid edges is
\begin{equation}\label{squares}
\begin{split}
^{10}\Psi^2_G=&4\Psi(G_0\backslash AB/4C)\Psi^{23,bc}(H_0\backslash4)\times\\
&\quad\big[\Psi(G_0\backslash4/ABC)\Psi^{23,bc}(H_0/4)-\Psi(G_0/4ABC)\Psi^{23,bc}(H_0\backslash4)\big].
\end{split}
\end{equation}
\end{lem}
\begin{proof}
We show that $g_A=g_B=g_C=0$ in Lemma \ref{leminiqred}, yielding $^9\Psi^2_G=-4f_0f_{ABC}g_0^2$.

By (\ref{cd}) the coefficient $g_A$ is given by the Dodgson $\Psi^{23,bc}(G\backslash A/1aBC)$.
Because $12BC$ is a square in $G$ the edge 2 is a self-loop in the graph $G\backslash A/1aBC$. The vanishing of $g_A$ then follows from Lemma \ref{lemcuts}.
Likewise $g_B=\Psi^{23,bc}(G\backslash B/1aAC)$. Now, $13AC$ is a square and the edge 3 is a self-loop in $G\backslash B/1aAC$.

To prove $g_C=0$ we use that by Theorem \ref{standarddenomred} the 6-invariant $\pm^6\Psi_G$ is invariant under permuting 1, 2, 3, $a$, $b$, $c$.
The graph $G_0=G\backslash12ab/3c$ in (\ref{56invG}) is trivially invariant. Solving (\ref{56inv}) for $\Psi^{23,bc}(G/1a)$ and permuting 1 and 2 gives
$\Psi^{23,bc}(G/1a)=\pm\Psi^{13,bc}(G/2a)$. Therefore, $g_C=\Psi^{13,bc}(G\backslash C/2aAB)$. With $23AB$ being a a square in $G$,
the edge 3 is a self-loop in $G\backslash C/2aAB$ and we get $g_C=0$ from Lemma \ref{lemcuts}.

Now we are in the situation of (\ref{case2}) in Definition \ref{defqdr}. In the graph $G_0\backslash AB/C$ edge 4 has a degree 1 vertex and by (\ref{cd}), (\ref{fdef}),
and Lemma \ref{lemcuts} we get that $f_0=\Psi(G_0\backslash AB/C)=\Psi(G_0\backslash AB/4C)$ is constant in $\alpha_4$.
This implies $D=0$ in (\ref{case2}) and the result follows from (\ref{cd}) and the explicit expressions for $f_{ABC}$ and $g_0$ in (\ref{fdef}) and Lemma \ref{leminiqred}.
\end{proof}

We did not find reductions beyond ten edges in prime ancestors with no triangles.
\begin{remark}
A referee explained the observation that the resultant on the right hand side of (\ref{squares}) factorizes in case (II):
The graph $H_0$ has a three-valent vertex $2cd$, and (\ref{v0a}) yields with $I=b$, $J=23$ and $e=c$ that
$$
\Psi^{23,bc}(H_0)=-\Psi^{23,b2}(H_0/c)-\Psi^{23,bd}(H_0/c).
$$
In the last term the edges 2 and 3 are a cycle, hence $\Psi^{23,bd}(H_0/c)=0$ by Lemma \ref{lemcuts}. With (\ref{psidef}) and (\ref{cd}) we conclude that
$\Psi^{23,bc}(H_0)=\Psi^{3,b}(H_0\backslash2/c)$. Moreover $G_0/ABC=H_0\backslash2b/3c$, so that the square bracket in (\ref{squares}) is
$$
\Psi^{b4,b4}_3\Psi^{3,b}_4-\Psi^{b,b}_{34}\Psi^{34,b4},
$$
evaluated on the graph $H_0\backslash2/c$. Identity (23) in Lemma 30 of \cite{Bperiods} is
$$
\Psi^{Ix,J}_S\Psi^{Iyz,Jz}_S-\Psi^{Ixz,Jz}_S\Psi^{Iy,J}_S=-\Psi^{Iz,J}_S\Psi^{Ixy,Jz}_S,
$$
where $I,J$ are words and $x,y,z\in IJ$, $S=IJ\cup\{x,y,z\}$. The signs are adjusted such that the special cases $x=y$ and $x=z$ are consistent.
(Note the anti-symmetry in $x,y$ on both sides.) With $I=\emptyset$, $J=x=b$, $y=3$, $z=4$, $S=\{3,4\}$ we get
$$
^{10}\Psi^2_G=4\Psi(G_0\backslash AB/4C)\Psi^{23,bc}(H_0\backslash4)\Psi^{4,b}_3(H_0\backslash2/c)\Psi^{3,4}(H_0\backslash2b/c).
$$
\end{remark}

\section{Affine reduction}\label{sectaffred}
Before calculating $(F)_q$ for a homogeneous polynomial $F$ (and a fixed odd prime power $q$) it is possible to reduce one variable by passing to affine coordinates.

\begin{lem}\label{lemred}
Let $q$ be an odd prime power and $F\in\ZZ[\alpha_1,\ldots,\alpha_N]$ be homogeneous of degree $2N$. Then
\begin{equation}\label{affred1}
(F)_q\equiv-(F|_{\alpha_1=1})_q-(F|_{\alpha_1=0,\alpha_2=1})_q+(F|_{\alpha_1=0,\alpha_2=0})_q\mod q.
\end{equation}
On the right hand side ambient spaces are $\FF_q^{N-1}$ and $\FF_q^{N-2}$, respectively.
\end{lem}
\begin{proof}
We split the sum over $\alpha_1$ in (\ref{Fqdef}) into $\alpha_1\neq0$ and $\alpha_1=0$. If $\alpha_1\neq0$ we scale $\alpha_i\mapsto\alpha_1\alpha_i$ for $i=2,\ldots,N$
providing $F\mapsto \alpha_1^{2N}F|_{\alpha_1=1}$. The factor $\alpha_1^{2N}$ is a non-zero square in $\FF_q$. By multiplicativity of the Legendre symbol this factor is trivial
and the sum over $\alpha_1$ provides the factor $q-1$. We obtain
\begin{equation}\label{a1}
(F)_q=(q-1)(F|_{\alpha_1=1})_q+(F|_{\alpha_1=0})_q.
\end{equation}
If we reduce the second term in the sum by the same method applied to $\alpha_2$ we get the result.
\end{proof}

Because of the smaller ambient space and the simplification from setting $\alpha_1=0$ in $F$ the two rightmost terms in (\ref{affred1}) are significantly faster to
count than $(F|_{\alpha_1=1})_q$. An even more powerful and mathematically nicer reduction is provided by the following theorem.

\begin{thm}\label{thmred}
Let $q=p^n$ be an odd prime power and $F\in\ZZ[\alpha_1,\ldots,\alpha_N]$ be homogeneous of degree $2N$.
Assume $F$ is of degree $\leq4$ in each variable. Let $F_i=F|_{\alpha_i=0}$ and $F^i=$ coefficient of $\alpha_i^4$ in $F$ (for $i=1,\ldots,N$). Then
\begin{equation}\label{affred2}
(F)_q\equiv-(F|_{\alpha_1=1})_q-\sum_{i=2}^N(F_1^i)_q\mod q
\end{equation}
with ambient spaces $\FF_q^{N-1}$ and $\FF_q^{N-2}$ (respectively) on the right hand side.
\end{thm}
The formula (\ref{affred2}) is powerful because the $F_1^i$ are amenable to quadratic denominator reduction and affine reduction. Moreover, one can choose a convenient edge 1.
This makes computing $(F_1^i)_q$ basically trivial.
\begin{proof}
We use the following setup: Define $F_I^J$ for sets of edges $I,J$ in analogy to $F_i$, $F^i$ in the theorem and write $f_I^J$ for $(F_I^J)_q\mod q$.
All equivalences are modulo $q$. Capital letters are sets and we write compositions $AB\cdots$ for the {\em disjoint} union of sets $A$, $B$, \ldots.
We prove that for disjoint subsets $I,J\subseteq X$ we have
$$
f_I^J\equiv(-1)^{|IJ|}\sum_{\genfrac{}{}{0pt}2{IJB\subseteq X}{|JB|=|I|}}f_I^{JB}.
$$
For $I=\{1\}$, $J=\emptyset$, and $X=\{1,2,\ldots,N\}$ we get (\ref{affred2}) from (\ref{a1}).

Note that the above formula is trivial for $|I|=|J|$. By Corollary \ref{corcw} we have $f_I^J\equiv0$ if $|I|<|J|$. We can hence restrict ourselves to the case $|I|>|J|$.

We use the notation $f_I^{J\times}$ for the object analogous to $f_I^J$ with the sums restricted to $\FF_q^\times$.
By inclusion-exclusion we have
$$
f_I^J\equiv\sum_{IJK\subseteq X}f_{IK}^{J\times},\quad\hbox{and }f_I^{J\times}\equiv\sum_{IJK\subseteq X}(-1)^{|K|}f_{IK}^{J}.
$$
We use the Cremona transformation $\alpha\mapsto\alpha^{-1}$ if $\alpha\in\FF_q^\times$. Concretely, we define a dual polynomial $\tilde{F}$ by mapping an exponent $k$ to
$4-k$ in every variable of the polynomial $F$. Likewise we define $\tilde{f}_I^J$ and $\tilde{f}_I^{J\times}$. In $\tilde{f}$ we swap the role of sub- and superscripts
so that taking sub- and superscripts commutes with dualizing. We have
$$
f_I^{J\times}\equiv\tilde{f}_I^{J\times},\quad\hbox{and }\tilde{f}_I^J\equiv0\hbox{ if }|J|<|I|,
$$
by Corollary \ref{corcw}. Double use of inclusion-exclusion gives (we map $f\mapsto f^\times\mapsto\tilde{f}^\times\mapsto\tilde{f}$)
$$
f_I^J\equiv\sum_{IJKL\subseteq X}(-1)^{|L|}\tilde{f}_{IK}^{JL},\quad \tilde{f}_I^J\equiv\sum_{IJKL\subseteq X}(-1)^{|L|}f_{IL}^{JK}.
$$
By Corollary \ref{corcw} we may restrict ourselves to the case $|JL|\geq|IK|$. Combining the two formulae yields
$$
f_I^J\equiv\sum_{IJKLST\subseteq X}(-1)^{|LS|}f_{IKS}^{JLT}.
$$
We fix $X$ and proceed by induction from larger to smaller superscripts. The cases where the superscript has $\geq|X|/2$ elements are trivial.
Because $|I|>|J|$ and $|JL|\geq|IK|$, non-trivial terms have non-empty $L$. The right hand side of the above equation has strictly larger superscripts than the left hand side.
By induction we get
$$
f_{IKS}^{JLT}\equiv\sum_{M:\genfrac{}{}{0pt}2{IJKLMST\subseteq X}{|IKS|=|JLMT|}}(-1)^{|M|}f_{IKS}^{JLMT}.
$$
Substitution into the previous equivalence gives
$$
f_I^J\equiv\sum_{\genfrac{}{}{0pt}2{IJAB\subseteq X}{|IA|=|JB|}}f_{IA}^{JB}\sum_{\genfrac{}{}{0pt}2{K\subseteq A}{LM\subseteq B}}(-1)^{|LM|+|A|-|K|}.
$$
In the rightmost sum only cardinalities matter. It can hence be written as
$$
\sum_{k=0}^{|A|}\sum_{\ell=0}^{|B|}\sum_{m=0}^{|B|-\ell}(-1)^{|A|-k+\ell+m}\binom{|A|}k\binom{|B|}\ell\binom{|B|-\ell}m.
$$
The sum over $m$ gives a Kronecker $\delta_{|B|,\ell}$. This simplifies the expression to
$$
\sum_{k=0}^{|A|}(-1)^{|A|+|B|-k}\binom{|A|}k=(-1)^{|B|}\delta_{|A|,0}.
$$
Substitution into the above formula for $f_I^J$ completes the proof.
\end{proof}

\section{Results}\label{sectres}

\renewcommand{\arraystretch}{1.2}
\setlength{\tabcolsep}{3pt}
\begin{table}
\begin{center}
\begin{tabular}{l|l||l|l||l|l}
$-c_2$&graph         &$-c_2$&graph         &$-c_2$&graph\\\hline
0&P[6,4]             &$[3,11]$&AN[9,976]   &$[4,23]$&AN[10,726]\\
1&P[3,1]             &$[3,12]$&P[9,157]    &$[4,32]$&AN[10,8209]\\
$(2/q)$&AN[10,35200] &$[3,15]$&AN[10,9803] &$[4,56]$&AN[10,5301]\\
$(-3/q)$&P[7,11]     &$[3,20]$&AN[10,3498] &$[4,73]$&P[11,7909]\\
$(4/q)$&P[7,8]       &$[3,24]$&P[11,7154]  &$[5,4]$&P[9,161]\\
$(-4/q)$&P[8,40]     &$[4,5]$&P[8,38]      &$[5,8]$&A[12,5129]\\
$(5/q)$&N[8,179]     &$[4,6]$&P[9,166]     &$[6,3]$&P[8,41]\\
$(9/q)$&A[13,5815]   &$[4,7]$&P[10,968]    &$[6,4]$&P[9,189]\\
$(-12/q)$&P[12,4363] &$[4,8]$&P[11,7146]   &$[6,7]$&P[9,173]\\
$[2,11]$&AN[9,958]   &$[4,10]$&P[11,7913]  &$[6,10]$&P[10,1170]\\
$[2,14]$&AN[9,646]   &$[4,12]$&A[12,5106]  &$[6,11]$&P[11,7175]\\
$[2,15]$&AN[9,638]   &$[4,13]$&P[9,167]    &$[6,17]$&A[10,8106]\\
$[2,17]$&AN[10,8071] &$[4,16]$&A[12,2354]  &$[7,3]$&P[9,188]\\
$[2,36]$&AN[10,21305]&$[4,17]$&P[10,959]   &$[8,2]$&P[10,1022]\\
$[2,37]$&AN[10,2011] &$[4,19]$&AN[10,698]  &$[8,5]$&P[10,1113]\\
$[3,7]$&P[8,37]      &$[4,21]$&P[11,7906]  &$[10,3]$&P[11,7710]\\
$[3,8]$&P[8,39]      &$[4,22]$&AN[10,27557]&&
\end{tabular}
\end{center}
\caption{First occurrences of identified $c_2$s. The notation $[w,x]$ refers to a weight $w$ level $x$ newform;
$X[\ell,n]$ refers to loop order $\ell$ and graph number $n$ in the file $X$ which is {\tt Periods.m} (P), {\tt ListOf$\ell$LoopAncestors.m} (A, which extends P to weights 12 and 13),
{\tt PeriodsNonPhi4} (N), {\tt ListOf$\ell$LoopNonPhi4Ancestors.m} (AN, which extends N to weights 9 and 10) in \cite{Shlog}.}
\label{tab2}
\end{table}

In this more experimental section we specialize to $q=p$. The target is to conjecturally identify Legendre symbols and modular forms in $c_2$s. The correspondence to modular forms
in Definition \ref{defmod} only uses primes $p$. For Legendre symbols one could test the results for small non-trivial prime powers 4, 8, 9, $\dots$.
We did not pursue this, but we assume that the results will be consistent with the identifications for primes (see Conjecture \ref{conqp}).

We determine initial prime sequences (prefixes) for the $c_2$-invariants of all $\phi^4$ graphs up to 11 loops and all non-$\phi^4$ graphs up to 10 loops. We also investigated
prime $\phi^4$ ancestors of loop orders 12 and 13 with hourglass subgraphs. In the case of loop order 12 we only determined their $c_2$ if quadratic denominator reduction led to a maximum
of 10 variables or if standard denominator reduction lead to Dodgson pairs in $\leq13$ variables (see Definition \ref{defdodpair}). For 13 loops we mostly restricted ourselves to the latter case.
In total we added 1731 $c_2$ calculations at 11 loops, $1749+1464$ calculations at 12 loops and $87+2044$ calculations at 13 loops to the results in \cite{BSmod}. In the non-$\phi^4$ case
we analyzed all 36247 graphs of ancestor type at 10 loops. In total we were able to distinguish 4801 sequences (in \cite{BSmod} we had 157 unique cases).

Table \ref{tab1} in the introduction gives all identified $c_2$s with superscripts indicating the maximum prime which was counted to establish the identification.
Note that some identifications are proved for all primes, either by a full quadratic denominator reduction for dimension zero or in \cite{K3,Lproofs}.
Because we do not expect any visual connection between the graph and the $c_2$ (see \cite{Yc2pre}) we do not plot the graphs of identified $c_2$-invariants.
Table \ref{tab2} gives their first occurrences with reference to \cite{Shlog} where also all known prefixes are available.

We identified two types of $c_2$s. Type 1 are Legendre symbols which correspond to dimension zero in Table \ref{tab1}.
Note that the {\em quasi constants} $z_2$, $z_3$, $z_4$, $y_5$ in \cite{BSmod} are the Legendre symbols $(4/q)$, $(-3/q)$, $(-4/q)$, $(5/q)$, respectively.

Type 2 are modular $c_2$s defined as follows.
\begin{defn}[Definition 21 in \cite{BSmod}]\label{defmod}
A primitive graph $G$ is modular if there exists a normalized Hecke eigenform $f$ for a congruence subgroup of $S\hspace{-1pt}L_2(\ZZ)$,
possibly with a non-trivial Dirichlet-character, with an integral Fourier expansion
$$
f(\tau)=\sum_{k=0}^\infty a_kq^k,\quad q=\exp(2\pi{\rm i}\tau),\quad a_k\in\ZZ,
$$
such that the $c_2$-invariant satisfies
$$
c_2^{(p)}(G)\equiv-a_p\mod p
$$
for all primes $p$. The dimension of the $c_2$ is the weight of $f$ minus 1.
\end{defn}
Fourier coefficients of modular forms were calculated with Sage \cite{Sage} which has a complete set of pre-built commands to do this.
In general, newforms of high and odd weights are harder to generate than low and even weights.
Modular $c_2$s were only searched up to the following generated levels.
\vskip1ex

\setlength{\tabcolsep}{2pt}
\renewcommand{\arraystretch}{1.1}
\noindent\begin{tabular}{c|c|c|c|c|c|c|c|c|c|c|c|c|c|c|c}
weight&2&3&4&5&6&7&8&9&10&11&12&13&14&15&16\\\hline
max.\ level&2000&500&1000&400&500&250&300&150&200&100&150&100&100&50&50
\end{tabular}
\vskip1ex
To unidentified $c_2$s we associate a dimension according to the following definition.

\begin{defn}\label{defdim}
A model of a non-zero $c_2$-invariant is a projective variety defined over $\ZZ$ of dimension $d$ whose point-count modulo $q=p^n$, $p=2,3,5,\ldots$ prime,
is $1-(-1)^dc_2^{(q)}$ modulo point-counts of varieties of strictly smaller dimension. A minimal model is a model of lowest dimension.
The dimension of a minimal model is the dimension $\dim(c_2)$ of the $c_2$-invariant.
\end{defn}
In general, it is almost impossible to prove that a $c_2$ has a certain dimension $d$ beyond the cases $d=0$ and $d=1$. Moreover, very little data exists for $q\neq  p$.
By modularity of algebraic varieties (related to the Langlands program) the above definition is expected to be consistent with the dimension in Definition \ref{defmod}
(see e.g.\ \cite{Schuett} and the references therein).
In view of Conjecture \ref{conqp} we can use the correspondence to conjecture the dimensions of modular $c_2$s (as in the top row of Table \ref{tab1}).
Standard denominator reduction provides upper bounds for the dimensions of non-identified $c_2$s. These bounds seem never to be sharp.
Quadratic denominator reduction may provide sharp bounds which can be only conjectural because the case $q=2^n$ is missing, see Lemma \ref{lemdim} and Conjectures \ref{dimbound1},
\ref{dimbound2}.

In general, it is hard to find varieties which correspond to given modular forms \cite{Schuett}.
Motivated by the structure of denominator reduction we defined Dodgson intersections in Definition \ref{defdodpair}. These are varieties which have a rather specific representation as
the intersection of the zero loci of two polynomials which are linear in each variable (the Dodgson pair), see Table \ref{tab3} for examples. Conjecture \ref{condodpair} suggests
that every minimal model is a Dodgson intersection. In particular, Dodgson pairs should give varieties for all modular forms which occur in $c_2$ invariants.
The simplicity of the examples in Table \ref{tab3} suggests that Dodgson pairs could be a good setup to search for varieties that correspond to some given modular forms.

In spite of the experimental flavor of this section, it is possible to prove two interesting negative results with the data of finite prefixes.
\begin{result}\label{result}
Assuming the completion conjecture (see Section \ref{sectc2}) we have the following statements for $\phi^4$ theory up to loop order 11.
\begin{enumerate}
\item There exists no weight two modular $c_2$-invariant of level $\leq2000$.
\item The only prime ancestor with $c_2^{(q)}\equiv-1\mod q$ for all prime powers $q$ is the complete graph with five vertices $K_5$.
\end{enumerate}
\end{result}
These results were previously found up to 10 loops in \cite{BSmod} (see Conjectures 26 and 25). We checked the consistency of our partial data at loop orders 12 and 13 with (2) and (1)
up to level 200. See Problem \ref{probgraphc2} for an attempt to interpret result (2).

\section{Conjectures and problems}\label{sectconj}
Section 5.1 in \cite{BSmod} handles the case of quasi-constant $c_2$. This case is superseded in this article by Legendre symbols in the sense that
every $c_2$ that is minus a Legendre symbol is quasi-constant and there conjecturally exist no quasi-constant $c_2$ invariants beyond negative Legendre symbols.
In the notation of \cite{BSmod} we have quasi-constants $0,1,z_2,z_3,z_4$ corresponding to the Legendre symbols $(0/q), (1/q), (4/q), (-3/q), (-4/q)$, respectively.
The first statement of Conjecture 25 in \cite{BSmod} claims that these are all quasi-constant $c_2$ invariants.
This is false as it misses the Legendre symbols $(-12/q)$ at 12 loops and $(9/q)$ at 13 loops, see Table \ref{tab1}.
We are reluctant to conjecture that with these additions the list is now complete. The second statement of Conjecture 25 in \cite{BSmod} has been confirmed in Result \ref{result} (2).

\begin{con}[Second statement of Conjecture 25 in \cite{BSmod}]\label{K5}
The only $\phi^4$ prime ancestor with $c_2^{(q)}\equiv-1\mod q$ for all $q$ is the complete graph $K_5$.
\end{con}

\begin{con}[Conjecture 26 in \cite{BSmod}]\label{nowt2}
If a primitive $\phi^4$ graph is modular with respect to a modular form $f$ then the weight of $f$ is $\geq3$.
\end{con}

The main motivation to study $c_2$-invariants comes from the demand to understand the geometries underlying perturbative quantum field theory.
From the existent data one may be let to the conclusion that there exists an analytical property which holds for all $c_2$s.

\begin{con}[Dodgson intersections]\label{condodpair}
Every non-zero $c_2$ has a minimal model which is a Dodgson intersection (see Definitions \ref{defdodpair} and \ref{defdim}).
\end{con}

In dimension zero the point-count of a Dodgson intersection is the negative Legendre symbol of a K\"all\'en function (\ref{Kaellen}) in the coefficients.
This gives all $(a/q)$ with $a\not\equiv3\mod4$. The first missing Legendre symbol in Table \ref{tab1} is $(-2/q)$. We, however, do not conjecture that
the point-counts of all Dodgson intersections should appear in the $c_2$ of some (non-$\phi^4$) graphs. It is an open question if there exists a graph beyond loop order 10
whose $c_2$ is $-(-2/q)$.

It is unclear which elliptic curves are Dodgson intersections. Experiments with coefficients $\{-1,0,1,2\}$ lead (via modularity) to levels
11, 14, 15, 17, 20, 21, 24, 26, 30, 32, 34, 36, 37, 38, 40, 42, 43, 48, 50, $\ldots$.
Newforms of the levels 19, 26, 27, 33, 35, 37, 38, 39, 44, 45, 46, 49, 50, $\ldots$ were not found (matching entries in both lists refer to the existence of two newforms at that level).

Note that standard denominator reduction gives Dodgson pairs if both factors are linear in each variable.
However, standard denominator reduction seems never to proceed to the dimension of the $c_2$ and quadratic denominator reduction does not give Dodgson pairs.
So, denominator reduction cannot be used as support for the above conjecture (see Problem \ref{higherred}). The only support for Conjecture \ref{condodpair} is Table \ref{tab3}
where we list examples of conjectural Dodgson pairs for some $c_2$s. We confirmed the point-counts for all primes from 2 to 97 but not for non-trivial prime powers (although this
is possible). We also did not prove minimality of the Dodgson pairs but relied on the conjectural relation between the weight of the modular form and the dimension of the corresponding
variety.

By Conjecture \ref{nowt2} (e.g.) we do not expect all Dodgson intersections to be perturbative geometries. At least one arithmetic property
seems necessary to describe geometries of $c_2$-invariants in $\phi^4$ theory. One might suggest the following picture:
\vskip1ex

\begin{center}
\begin{tabular}{l|ll}
&Dodgson intersection&Arithmetic property\\\hline
$\phi^4$&Yes&Yes\\
non-$\phi^4$&Yes&No
\end{tabular}
\end{center}
\vskip1ex

In zero dimensions the arithmetic property seems to be related to exceptional primes 2 and 3. In one dimension the arithmetic property should be strong
enough to eliminate all elliptic curves. Regretfully, we have no conjecture for this arithmetic property.

\setlength{\tabcolsep}{3pt}
\begin{table}
\begin{center}
\begin{tabular}{l|ll}
$-c_2$&Dodgson pair&\\\hline
1&$x$&$yz$\\
$(2/q)$&$x+y-z$&$xy+xz+2yz$\\
$(-3/q)$&$x+y+z$&$xy+xz+yz$\\
$(4/q)$&$x+y$&$xz-yz$\\
$(-4/q)$&$x+y+2z$&$xy+xz+yz$\\
$(5/q)$&$x+y-z$&$xy+xz+yz$\\
$(9/q)$&$2x+y$&$xz-yz$\\
$(-12/q)$&$2x+2y+2z$&$xy+xz+yz$\\
$[2,11]$&$wx+wz+yz$&$wx-wy+xy+xz$\\
$[2,14]$&$wx+wz+yz$&$wy-wz+xy+xz$\\
$[2,15]$&$wx+wz+yz$&$wy+xy+xz$\\
$[2,17]$&$wx-wz-yz$&$wy+xy+xz$\\
$[2,36]$&$wx+wz+yz$&$wy-wz+2xy+xz$\\
$[2,37]$&$wx-wz+yz$&$wx+wy+xy-xz$\\
$[3,7]$&$vw+vx+yz$&$vwy+vwz+wxy+wxz+xyz$\\
$[3,8]$&$vw+vx+yz$&$vwy+vxz-wxy-wxz$\\
$[3,11]$&$vw+2vx+4yz$&$vwy+vwz+vyz+wxy+wxz$\\
$[3,12]$&$vw+vx+yz$&$4vyz+wxy+wxz$\\
$[3,15]$&$vw+vx+vy+wz$&$vwz+vxy-wxy+wxz+xyz$\\
$[3,20]$&$vw+vx+wy+xz$&$vwz-vxy+vyz+wxy-wxz$\\
$[3,24]$&$2vw+vx+2vy+wz$&$vwz+2vxz+2wxy+2xyz$\\
$[4,5]$&$uvx+uxy+uxz+vwy+vwz$&$uwz+uyz+vwx+vwz+vxy$\\
&&\quad$+vyz+wxy+wyz$\\
$[4,6]$&$uvx+uxy+uxz+vwz$&$uvy+uwy+uwz+vwx+vwz$\\
&&\quad$+vxy+vyz+wxy+wyz$
\end{tabular}
\end{center}
\caption{Examples of (not unique) Dodgson pairs for some $c_2$s, confirmed for primes up to 97. Conjecturally these Dodgson pairs are minimal models in the sense of Definition \ref{defdim} and
Conjecture \ref{condodpair}. The first column refers to $-c_2$ for Legendre symbols or newforms given as [weight, level].}
\label{tab3}
\end{table}

\begin{problem}[The main problem]\label{probgeo}
Find the missing arithmetic property so that one can conjecturally describe the class of perturbative geometries.
\end{problem}

In general, dimensions of non-identified $c_2$-invariants are hard to guess. For odd primes, however, have the following lemma.

\begin{lem}\label{lemdim}
Let $G$ be a graph with $N=2h_1(G)$ edges. If the quadratic $n$-invariant $^n\Psi^2_G(\alpha_{n+1},\ldots,\alpha_N)$ exists for some $n<N$ then for odd prime powers
the zero locus of $\alpha_0^2\alpha_N^{2(N-n-1)}-{}^n\Psi^2_G(\alpha)$ is a model of dimension $N-n-1$ for $c_2(G)$.
\end{lem}
\begin{proof}
By (\ref{pc2ls}) with $\alpha_0$ for the variable $\alpha$ we get
$$
\big[\alpha_0^2\alpha_N^{2(N-n-1)}-{}^n\Psi^2_G(\alpha)\big]_q\equiv\big(4\alpha_N^{2(N-n-1)}\,{}^n\Psi^2_G(\alpha)\big)_q+\big(\alpha_N^{2(N-n-1)}\big)_q\mod q.
$$
If $n=N-1$ then $(1)_q=q$ yielding $[\alpha_0^2-{}^n\Psi^2_G(\alpha_N)]_q\equiv({}^n\Psi^2_G(\alpha_N)\big)_q$.
If $n<N-1$, the second term is $(q-1)q^{N-n-1}\equiv0\mod q$. We may restrict the first term to $\alpha_N\neq0$ and obtain with $(4\alpha_N^{2(N-n-1)}/q)=1$,
\begin{equation}\label{dim1}
\big(^n\Psi^2_G\big)_q\equiv\big[\alpha_0^2\alpha_N^{2(N-n-1)}-{}^n\Psi^2_G(\alpha)\big]_q+\big(^n\Psi^2_G|_{\alpha_N=0}\big)_q\mod q.
\end{equation}
Iterating (\ref{dim1}) gives $(^n\Psi^2_G)_q$ as a sum of affine point-counts of decreasing dimensions.
The projective point-count is $([\bullet]_q-1)/(q-1)\equiv1-[\bullet]_q\mod q$. From Theorem \ref{thmquadratic} we get $c_2^{(q)}(G)\equiv(-1)^{n-1}(^n\Psi^2_G)_q\mod q$.
Because the point-count in (\ref{dim1}) has projective dimension $d=N-n-1\equiv n-1\mod 2$ the result follows.
\end{proof}

If the point-count of $\alpha_0^2\alpha_N^{2(N-n-1)}-{}^n\Psi^2_G(\alpha)$ fails to reproduce $c_2(G)$ for the prime 2 then one can typically manipulate the polynomial
by scaling some variables $\alpha_i\to2\alpha_i$ such that the point-count also gives $c_2^{(2)}(G)$. With Conjecture \ref{conqp} for $p=2$ we expect the following statement.

\begin{con}\label{dimbound1}
Let $G$ be a primitive graph for which the quadratic $n$-invariant $^n\Psi^2_G$ exists. Then $\dim(c_2(G))\leq2h_1(G)-n-1$.
\end{con}

If $^n\Psi^2_G$ factorizes into homogeneous polynomials $A$, $B$ of degree $N-n-1$ and $N-n+1$, respectively, we obtain from (\ref{pc2ls}), Corollary \ref{corcw}, and Theorem
\ref{thmquadratic} that $[\alpha_0^2A-B]_q\equiv(-1)^{n-1}c_2^{(q)}(G)\mod q$ for odd prime powers $q$. This factorization is provided by the 10-invariant (\ref{squares})
leading to a point-count polynomial which is quadratic in each variable. The zero locus of such a polynomial can be a Dodgson intersection of dimension $2h_1(G)-11$.
The general 9-invariant in Lemma \ref{leminiqred} does not factorize accordingly.
We expect that there always exists a 10th reduction which involves the structure of the fourth 3-valent vertex in the graph $G$ (see Problem \ref{higherred}).

\begin{con}\label{dimbound2}
Any primitive graph $G$ has $\dim(c_2(G))\leq2h_1(G)-11$.
\end{con}

By the conjectural correspondence between modular forms and varieties the above bound should be sharp at loop order 8 for $P_{8,41}$ (in the notation of \cite{Scensus})
whose $c_2$ is conjecturally associated to the weight 6 level 3 modular form, see Table \ref{tab1} and \cite{BSmod}.

In some cases quadratic denominator reduction does go all the way down to $\dim(c_2)$. The minimum number of variables after quadratic denominator reduction
in the case of {\em unidentified} $c_2$-invariants in $\phi^4$ is five. This conjecturally refers to dimension $\leq4$ (one particular case of a four-dimensional
variety with unidentified $c_2$ was given in Section 5.3 in \cite{BSmod}). We conjecture that four is a sharp lower bound for the dimension of unidentified $c_2$-invariants in $\phi^4$.
Unidentified sequences at dimensions $\geq4$ also seem to account for the noticeable drop of modular forms at weights $\geq5$ in Table \ref{tab1}.

For non-$\phi^4$ graphs we have quadratic denominator reductions to four variables in unidentified sequences.
We hence seem to have first unidentified sequences at dimension 3 in non-$\phi^4$.

\begin{con}\label{condim23}
Every $\phi^4$ $c_2$-invariant of dimension 2 or 3 is modular. No $\phi^4$ $c_2$-invariant has dimension 1.
Every non-$\phi^4$ $c_2$-invariant of dimension 1 or 2 is modular.
\end{con}

We also conjecture that there only exists a finite number of perturbative geometries at any given dimension.

\begin{con}\label{condimfinite}
The number of $\phi^4$ $c_2$-invariants at fixed dimension is finite.
\end{con}
The conjecture follows a general philosophy of sparsity in the geometries of $\phi^4$ periods. Experimentally it is only supported by Conjecture \ref{condim23}
and by not finding more than the listed Legendre symbols and modular forms in Table \ref{tab1}.

Note that most conjectures in this section are rather weakly supported.
We close the article with a series of open problems.

\begin{problem}[Dimension and minimal models]\label{probdim}
Given a graph or some Dodgson pair, (conjecturally) determine $\dim(c_2)$.
If $\dim(c_2)$ is (conjecturally) known, determine a minimal model.
\end{problem}
Because a general method to determine $\dim(c_2)$ seems out of reach, a conjectural result would already be a major achievement.

Once the dimension of a $c_2$ is (conjecturally) known one may hope to find a minimal model. One approach in this direction could be to find more improvements in
denominator reduction. One can expect that quadratic denominator reduction is optimal with respect to reductions in one variable.
The situation may be different if one considers two (or more) variables $\alpha_{n+1},\alpha_{n+2}$ and studies the structure of the denominator in these variables.
In this setup it is possible that there exist configurations which allow one to do more reductions. Ideally, one would like to have a denominator reduction that goes down to a minimal model.

\begin{problem}[Higher dimensional denominator reduction]\label{higherred}
Generalize denominator reduction by considering structures in two or more remaining variables.
Try to find a complete set of reductions, so that denominator reduction always gives a (conjecturally) minimal model.
\end{problem}
Note that in \cite{K3} such a two parameter reduction was used to get the minimal K3 model.

Another problem related to denominator reduction is to identify graphical configurations that give reductions. For the first reductions this was done in Section \ref{sectini}.
Note that these graphical reductions are very helpful at high loop orders because they reduce time- and memory-consuming factorizations of large polynomials.
Another interesting approach in this direction is to identify families of graphs which reduce to a fixed (small) number of variables, see \cite{Hourglass}.

\begin{problem}[Graphical denominator reduction]
Try to find explicit expressions in terms of Dodgson polynomials for a maximum number of denominator reductions.
\end{problem}

Conjecture \ref{K5} means that the family of primitive $\phi^4$ graphs which have $c_2$-invariant $-1$ would be completely characterized by the purely graph theoretical
property of having a $K_5$ ancestor. (The $K_5$-family is studied in \cite{LMYK5}.)
Note that there exist many primitive non-$\phi^4$ graphs which have $c_2\equiv-1$ and do not reduce to $K_5$ by double triangle reductions.
The $c_2\equiv-1$-family seems to be more or less the only family (defined by its $c_2$) which is fully characterized by the ancestor. It would be interesting to see if a more general
concept of ancestor can be defined which characterizes other $c_2$-families. See Section 4.6 in \cite{BYc2} for some results in this direction.

\begin{problem}[Graphical $c_2$-families]\label{probgraphc2}
Find more (ideally, a full set of) operations which reduce the graph while leaving the $c_2$ unchanged.
Graphically describe full families of graphs with equal $c_2$ (beyond the $K_5$-family).
\end{problem}

\appendix
\section{A Chevalley-Warning-Ax theorem for double covers of affine space\\[1ex]
by Friedrich Knop}
\subsection{Main Theorem}
In this appendix we prove the following theorem.

\begin{thm}\label{thm:main}
  Let $k$ be a finite field of characteristic $p\ne2$ and let
  $f\in k[x_1,\ldots,x_n]$ be a polynomial with $n\geq1$ and $\deg f<2n$. Let
  $X(k)\subseteq k^{n+1}$ be the set of $k$-valued solutions of
  \[\label{eq:eq}
    y^2=f(x_1,\ldots,x_n).
  \]
  Then the cardinality $|X(k)|$ is divisible by $|k|$.
\end{thm}

\subsection{Proof}

The following proof was suggested by Jason Starr in answer to a
question by F. Knop on mathoverflow.

The point of departure is a deep theorem of Esnault to the effect that
under a certain condition (translating into our bound of $\deg f$) the
number of $k$-valued points of a smooth projective $k$-variety $X$ is
congruent to $1$ modulo $|k|$.

Unfortunately, our $X$ is neither smooth nor projective. The latter
problem is solved by constructing a completion $\Xq$ of $X$ and
applying Esnaults theorem to $\Xq$.

The first problem is solved by realizing $\Xq$ as fiber of a
projective morphism $\pi$ between two smooth quasi-projective
varieties such that the assumptions of Esnault's theorem hold for the
generic fiber. Then a theorem of Fakhruddin-Rajan asserts that it
holds for every fiber and we are done.

The construction of $\Xq$ is classical. Put $d:=2n$ and let

\[
  \Fq(x_0,\ldots,x_n):=x_0^d\,f(\frac{x_1}{x_0},\ldots,\frac{x_n}{x_0})
\]
be the homogenization of $f$. Since $\deg f<d$, this is indeed a
homogeneous polynomial of degree $d$. We think of $\Fq$ as a section
of the line bundle $\cO_{\P^n}(d)$. Since $d=2n$ it equals
the square of $\cL:=\cO(n)$. Let $p:L\to\P^n$ be the geometric
realization of $\cL$, i.e.,
$L=\underline{\mathrm{Spec}}_{\cO_{\P^n}}\bigoplus_{\nu=0}^\infty\cL^{-\nu}$. Then
$p^*\cL$ comes with a tautological section $\sigma$ in degree $1$. Thus
both $\sigma^2$ and $p^*\Fq$ are sections of
$p^*\cL^2=p^*\cO_{\P^n}(d)$ and we define $\Xq\subseteq L$ as
the vanishing scheme of $\sigma^2-p^*\Fq$. It comes with an induced
projection $p:\Xq\to\P^n$. An immediate local calculation shows that
over $\A^n\subset\P^n$ the equation $\sigma^2=p^*\Fq$ becomes
\eqref{eq:eq}. Similarly, over the other affine charts of $\P^n$ one
gets an equation similar to \eqref{eq:eq}. This shows that
$p:\Xq\to\P^n$ is a finite flat morphism of degree $2$ which is
ramified over the vanishing set $R\subseteq\P^n$ of $\Fq$ (even for
$\Fq\equiv 0$ when $R=\P^n$). Moreover, one sees that $\Xq$ is smooth
if and only if $R$ is smooth.

Since $R$ is not smooth we now do the construction simultaneously
for all $\Fq$. For this, let $\fY$ be the space of all homogeneous
polynomials of degree $d$ which is an affine space of dimension
$N=\genfrac(){0pt}{}{d+n-1}{n-1}$. Then there is a universal homogeneous
polynomial of degree $d$
\[
  F(x_0,\ldots,x_n;a_{i_0\ldots i_n})=\sum_{i_0+\ldots+
    i_n=d}a_{i_0\ldots i_n}x_0^{i_0}\ldots x_n^{i_n}
  \in\cO(\A^n\times\fY).
\]
It can be considered as a section of the line bundle
$\cO_{\P^n\times\fY}(d)$. Let $\fL$ be the geometric realization of
$\cO_{\P^n\times\fY}(n)$ with tautological section $\sigma_\fY$. Then
we define $\fX\subseteq\fL\times\fY$ by the equation $\sigma_\fY^2=p^*F$. It
comes with a projection $\pi:\fX\to\fY$. There are three points to this
construction:

\begin{itemize}
\item The morphism $\pi$ is projective since it factors as
  $\fX\into\P^n\times\fY\auf\fY$ and if we think of $\Fq$ as an
  element of $\fY$ then the fiber $\pi^{-1}(\Fq)$ is just $\Xq$ from
  above. In particular, the morphism is surjective.

\item Both varieties $\fX$ and $\fY$ are smooth and irreducible. This
  is trivial for $\fY$. For the smoothness of $\fX$ observe that in
  the equation for $\fX$ over $\A^n\times\fY$
  \[
    y^2=F(1,x_1,\ldots,x_n,a_{i_0\ldots i_n})
  \]
  one can can solve for the constant term $a_{d0\ldots0}$.
 Hence, the equation defines an affine space of dimension $n+N-1$.

\item By the well-known theorem of Bertini, the generic fiber of $\pi$
  has a smooth ramification divisor $R$ and is therefore itself smooth (since ${\rm char} k\ne2$).

\end{itemize}

So we have achieved our goal to realize $\Xq$ as a fiber of a morphism
between smooth varieties whose generic fiber is smooth, as well. Now
we use a couple of deep theorems to prove the following
homogeneous version of Theorem \ref{thm:main}:

\begin{thm}\label{thm:homog}
  Let $k$ be a finite field of characteristic $p\ne2$ and let
  $\Fq\in k[x_0,\ldots,x_n]$ be a homogeneous polynomial with $n\geq1$ and
  $\deg\Fq\le 2n$. Let $\Xq$ be the double cover of $\P^n$ which is
  ramified over the zero set of $\Fq$. Then $|\Xq(k)|\equiv1\mod |k|$.
\end{thm}

\begin{proof}
  Generalizing a theorem of Esnault, \cite{Esnault}, Cor.~1.2, it
  was shown by Fakhruddin and Rajan, \cite{F-R}, Cor.~1.2, that every
  fiber of $\pi:\fX\to\fY$ satisfies
  $|\pi^{-1}(y)(k)|\equiv 1 \mod |k|$ as soon as $\fX$ and $\fY$ are
  smooth, $\pi$ is projective and surjective and
  $CH_0(Z)\otimes\QQ =\QQ$ where $CH_0$ denotes the $0$-th Chow
  group and $Z$ is the base change of the generic fiber of $\pi$
  to the algebraic closure of $k(\fX)$.

  Only the latter condition has still to be shown. It is for example
  satisfied if $Z$ is rationally chain connected, i.e., any two points
  of $Z$ can be connected by a connected chain of rational
  curves. This is because all points on a rational curve are (by
  definition) equivalent in $CH_0$.

  A Fano variety is a smooth projective variety $Z$ such that the
  inverse $-K_Z$ of the canonical divisor is ample. In our case $Z$ is
  a generic double cover of $\P^n$, hence smooth and projective. Its
  canonical divisor is calculated according to the formula
  \[
    2K_Z=p^*(2K_{\P^n}+R)
  \]
  where $R\subset\P^n$ is the ramification divisor of $p$, i.e., the
  zero set of $\Fq$ (Riemann-Hurwitz, see, e.g., \cite{Hartshorne}, Props.\ IV 2.2(c), 2.3).
  Let $D\subset\P^n$ be a hyperplane. Then
  $K_{\P^n}\sim-(n+1)D$ and $R\sim dD$. Hence $Z$ is Fano if and only
  if $2(n+1)-d>0$, i.e., $d\le2n$ since $d$ is even. So $Z$ is indeed
  Fano.

  Now we conclude with a theorem independently proved by Campana,
  \cite{Campana}, Cor.~3.2, and Kollar-Miyaoka-Mori,
  \cite{KollarMM}, Thm.~3.3, to the effect that $Z$ Fano implies that
  $Z$ is rationally chain connected.  
\end{proof}

Now Theorem \ref{thm:main} follows immediately: Since the degree of $f$ is strictly less than $d$,
its degree $d$ homogenization $\Fq$ is divisible by $x_0$. This means that $p:\Xq\rightarrow\P^n$ is ramified
over the hyperplane $\P^{n-1}$ at infinity. Therefore $\partial X:=p^{-1}(\P^{n-1})\rightarrow\P^{n-1}$ is bijective. Hence
$|\partial X(k)|=|\P^{n-1}(k)|=1+|k|+\ldots+|k|^{n-1}\equiv 1\mod |k|$. Since also $|\Xq(k)|\equiv1\mod |k|$
by Theorem \ref{thm:homog} we get $|X(k)|=|\Xq(k)|-|\partial X(k)|\equiv0\mod|k|$ as claimed.

\bibliographystyle{plain}
\renewcommand\refname{References}

\end{document}